%% file: Pairing-s1.tex
\documentclass[pra,amsmath,amssymb,amsfonts,twocolumn,nofootinbib]{revtex4}

\usepackage{bbm}

\usepackage{amsmath}
\usepackage{amssymb}
\usepackage{dsfont}
\usepackage{exscale}
\usepackage{graphicx}
\usepackage{epsfig}
\usepackage{amsthm}

\newtheorem{theorem}{Theorem}[section]

\newtheorem{lemma}[theorem]{Lemma}
\newtheorem{definition}[theorem]{Definition}

\newcommand{\ad}{a^{\dagger}}
\newcommand{\one}{\mathds{1}}
\newcommand{\tr}{\mbox{tr}}
\newcommand{\spec}{\mbox{spec}}
\newcommand{\bracket}{\rangle \langle}

\newcommand{\Pd}{P^{\dagger}}
\newcommand{\Qd}{Q^{\dagger}}

\newcommand{\Na}{\hat{N}_a}

\newcommand{\pdm}{p^{\dagger}_M}
\newcommand{\ppm}{p_M}
\newcommand{\qdm}{q^{\dagger}_M}
\newcommand{\qm}{q_M}
\newcommand{\Nb}{\hat{N}_b}

\newcommand{\vac}{|0\rangle}
\newcommand{\bn}{\bar{N}}
\newcommand{\sn}{\sigma_{\bar{N}}}

\newcommand{\cA}{\mathcal{A}}
\newcommand{\cS}{\mathcal{S}}

\renewcommand{\mathds}[1]{\mathbbm{#1}}
\newcommand{\ket}[1]{\left| #1\right>}

\newcommand{\ketbra}[2]{\vert #1\rangle \! \langle #2\vert}
\newcommand{\proj}[1]{\ketbra{#1}{#1}}

\begin{document}

\title{Pairing in fermionic systems: A quantum information perspective
}

\author{Christina V. Kraus$^1$, Michael M. Wolf$^{1,2}$, J. Ignacio Cirac$^1$
  and G\'eza Giedke$^1$}
\affiliation{(1) Max-Planck-Institute for Quantum Optics,
Hans-Kopfermann-Str.\ 1, D-85748 Garching, Germany.\\
(2) Niels-Bohr--Institute, Copenhagen University, Blegdamsvej 17, DK-2100
Copenhagen \O, Denmark}

\begin{abstract}
  The notion of ''paired'' fermions is central to important condensed matter
  phenomena such as superconductivity and superfluidity. While the concept is
  widely used and its physical meaning is clear there exists no
  systematic and mathematical theory of pairing which would allow
  to unambiguously characterize and systematically detect paired states. We
  propose a definition of pairing and develop methods for its detection and
  quantification applicable to current experimental setups.  Pairing is shown
  to be a quantum correlation different from entanglement, giving further
  understanding in the structure of highly correlated quantum systems. In
  addition, we will show the resource character of paired states for precision
  metrology, proving that the BCS states allow phase measurements at the
  Heisenberg limit.
\end{abstract}

\date{\today}

\maketitle
\input{Intro-s1}
\input{Classification-s1}
\input{Definition-s1}
\input{Gaussian-Pairing-s1}
\input{Number-conserving-s1}
\input{Interferometer-s1}
\input{Conclusion-s1}

\begin{appendix}
\input{Appendix-s1}

\end{appendix}

\bibliography{Pairing-s1}
\end{document}

%% file: Intro-s1.tex
\section{Introduction}
The notion of pairing in fermionic systems is at least as old as the seminal
work of Bardeen, Cooper and Schrieffer explaining superconductivity
\cite{PhysRev.108.1175}. The formation of fermionic pairs with opposite spin
and momentum is not only the source for the vanishing resistance in solid
state systems, but it can also explain many other interesting phenomena, like
superfluidity in helium-3 or inside a neutron star.

For instance, with recent progress in the field of ultracold
quantum gases fermionic pairing has gained again a lot of
attention \cite{ReaglmBEC, JochimmBEC,
  PhysRevLett.91.250401,Vortices, SF-direct,partridge:020404,
  zwierlein:120403, regal:040403}. These experiments allow an excellent
control over many parameters inherent to the system, offering a unique testing
ground for existing theories and an exploration of new and exotic phases.
However, the notion of pairing in these systems is less clear and sometimes
even controversial. Recent experiments on the BEC-BCS crossover have caused a
heated debate whether or not the obtained data was in agreement with pairing
\cite{partridge:020404,partridge_27, zwierlein_response,
  patridge_response}. In addition, pairing without superfluidity
\cite{Pairing-no-SF} has been observed in these experiments, raising
fundamental questions on quantum correlations in fermionic many-body systems.

Motivated by these exciting experiments and the central role
pairing plays in many physical phenomena, and by the perceived
lack of accepted criteria to verify the presence of pairing in a
quantum state, we propose a clear and unambiguous definition of
pairing intended to capture its two-particle nature and to allow a
systematic study of the set of paired states and its properties.
We employ methods and tools from quantum information theory to
gain a better understanding of the set of fermionic states that
display pairing. In particular, we develop tools for the
systematic detection and for the quantification of pairing, which
are applicable to current experiments. Our approach is inspired by
concepts and methods from entanglement theory, thus building a
bridge between quantum information science and condensed matter
physics.

Since they contain non-trivial quantum correlations, paired states belong to
the set of entangled many-body states. However, pairing will turn out to be
not equivalent to any known concept of entanglement in systems of
indistinguishable particles \cite{PhysRevA.65.042101,
  ZanardiJP,larsson:042320, wolf:010404, gioev:100503, cramer:220603,
  PhysRevA.67.024301, PhysRevA.64.042310, banuls:022311,
  PhysRevLett.91.097902,PhysRevA.64.022303, ESBL-Indistinguishable,
  dowling:052323,PhysRevB.63.085311, PhysRevA.70.042310,PhysRevLett.92.087904}
but to represent a particular type of quantum correlation of its
own. We will show that these correlations can be exploited for
quantum phase estimation. Hence pairing constitutes a resource in
state estimation using fermions as much as entangled states with
spins.

This article is organized as follows: After the introduction of
the language necessary for the description of fermionic systems in
Sec. \ref{Classification-section}, we will introduce the general
framework of pairing theory in Sec. \ref{Definition-Section}. This
part includes our definition of pairing and methods for its
detection and quantification. In order to fill the theory with
life we will apply it to two different classes of fermionic states
in Secs. \ref{Gaussian-Pairing-section} and
\ref{Number-conserving-section}. We start out with pairing in
fermionic Gaussian states in Sec. \ref{Gaussian-Pairing-section}.
The interest in this family of states is two-fold. First, the
pairing problem can be solved completely in this case, so that
Gaussian states are particularly interesting from a conceptual
point of view. Second, there exists a relation between pure
fermionic Gaussian states and the BCS states of superconductivity
(see Sec. \ref{Connection-N-NN} for the details) which are
examples of paired states par excellence. This enables us to
translate methods developed for the detection and quantification
of pairing for Gaussian states to the BCS-states. The reader
interested in the application of our pairing theory to
experimental application is referred to Sec.
\ref{Number-conserving-section}. There we study pairing for number
conserving states, i.e. states commuting with the number operator.
This class includes the states appearing in the BEC-BCS crossover,
and we will develop tools for the detection of pairing tailored
for these systems. In \ref{Interferometer-section} we will show
that certain classes of paired states constitute a resource for
quantum phase estimation, proving that pairing is a resource
similar to entanglement.

%% file: Classification-s1.tex
\section{Fermionic states}\label{Classification-section}
In this chapter we review the basic concepts needed for the
understanding of fermionic systems. We start out with some
notation used for the description of fermionic systems in second
quantization in Sec. \ref{Notation-section}. As pairing is a
special sort of correlation, we continue with a review on quantum
correlations and entanglement in systems of indistinguishable
particles in Sec. \ref{Correlation-section}. This general part is
followed by the introduction of fermionic Gaussian states and
number conserving states in Secs. \ref{Gaussian-section} and
\ref{Connection-N-NN}. The latter includes the introduction of
BCS-states and their relation to the Gaussian states. As this part
is only necessary for the application of the pairing theory to
these concrete examples in Secs. \ref{Gaussian-Pairing-section}
and \ref{Number-conserving-section}, it is possible to skip this
part at the beginning, and then refer to it later on.

\subsection{Basic notation}\label{Notation-section}
We consider fermions on an $M$-dimensional single particle Hilbert
space $\mathcal{H}=\mathds{C}^M$. All observables are generated by
the creation and annihilation operators $\ad_j$ and $a_j, j=1,
\ldots, M$, which satisfy the canonical anti-commutation relations
(CAR) $\{a_k, a_l\}=0$ and $\{a_k, \ad_l\}=\delta_{kl}$. We say
$\ad_j$ creates a particle in mode (or single particle state)
$e_j$, where $\{e_j\}\subset\mathds{C}^M$ denotes the canonical
orthonormal basis of $\mathcal{H}$. In general, for any normalized
$f\in\mathcal{H}$, we define $a_f \equiv \sum_kf_ja_j$, the
annihilation operator for mode $f$.

Sometimes a description using the $2M$ hermitian Majorana operators $c_{2j-1}=
\ad_j+a_j$, $ c_{2j}=(-i)(\ad_j-a_j)$, which satisfy $\{c_k,c_l\}=2
\delta_{kl}$, is more convenient.

The Hilbert space of the many-body system, the antisymmetric Fock
space over $M$ modes, $\mathcal{A}_M$, is spanned by the
orthonormal Fock basis defined by
\begin{equation}\label{eq:fockbasis}
|n_1, \ldots, n_M \rangle = \left(\ad_1 \right)^{n_1} \ldots
\left(\ad_M \right)^{n_M} |0 \rangle,
\end{equation}
where the vacuum state $|0 \rangle$ fulfills $a_j |0\rangle =0\; \forall j$.
The $n_j \in \{0,1\}$ are the eigenvalues of the the mode occupation number
operators $n_j = \ad_j a_j$.  The $N$-particle subspace spanned by vectors of
the form (\ref{eq:fockbasis}) satisfying $\sum_in_i=N$ is denoted by
$\mathcal{A}_M^{(N)}$. The set of density operators on the Hilbert space
${\cal H}=\mathcal{A}_M,\mathcal{A}_M^{(N)}$ is
denoted by $\cS({\cal H})$.

Linear transformations of the fermionic operators which preserve
the CAR are called canonical transformations. They are of the form
$c_k\mapsto c_k^{\prime}= \sum_i O_{kl} c_l,$ where $O \in O(2M)$
is an element of the real orthogonal group. These transformations
can be implemented by unitary operations $U_O$ on $\cA_M$ which
are (for $\det O=1$) generated by quadratic Hamiltonians in the
$c_j$ (see, e.g. \cite{LinearOptics}). The subclass of canonical
operations which commute with the total particle number
$N_\mathrm{op}=\sum_in_i$ are called passive transformations. They
take a particularly simple form in the complex representation $a_k
\mapsto a_k^{\prime}=\sum_l U_{kl}a_l$, where $U$ is unitary on
the single-particle Hilbert space $\mathcal{H}$, i.e., they
describe (quasi)free time-evolution of independent particles.
Canonical transformations which do not commute with
$N_\mathrm{op}$ are called active. They mix creation and
annihilation operators.

\subsection{Quantum correlations of fermionic
states}\label{Correlation-section} The notion of ''pairing''  used
in the description of superconducting solids, superfluid liquids,
baryons in nuclei, etc. is always associated with a correlated
fermionic system. The subject of quantum correlations in fermionic
systems is vast \cite{Mahan00}. In recent years, there has been
renewed interest from the perspective of quantum information
theory. There quantum correlations (aka entanglement) of
distinguishable systems (qubits) play a crucial role as a resource
enabling certain state transformations or information processing
tasks. The detailed quantitative analysis of quantum correlations
motivated by this has proven to be valuable also in the
understanding of condensed matter systems (see \cite{amico:517}
for a review).

In contrast to the usual quantum information setting which studies the
entanglement of distinguishable particles, the indistinguishable nature of the
fermions is of utmost importance in the settings of our interest. The existing
concepts for categorizing entanglement in systems of indistinguishable
particles fall into two big classes: Entanglement of modes
\cite{PhysRevA.65.042101,ZanardiJP,larsson:042320,wolf:010404,gioev:100503,cramer:220603,PhysRevA.67.024301,PhysRevA.64.042310,banuls:022311}
and entanglement of particles. Entanglement of particles has been considered
e.g. in \cite{PhysRevLett.91.097902, PhysRevA.64.022303,
  ESBL-Indistinguishable, dowling:052323, PhysRevB.63.085311,
  PhysRevA.70.042310, PhysRevLett.92.087904}, leading to the concept of Slater
rank \cite{PhysRevA.64.022303, ESBL-Indistinguishable}, being the
generalization of the Schmidt rank to indistinguishable particles. We show in
Sec.  \ref{Definition-Section} that our definition of pairing does not
coincide with any of the existing ideas. We refrain from giving an exhaustive
review on the existing concepts referring the interested reader to the
mentioned literature and references therein, and restrict to the following
definition:
\begin{definition}\label{DefSep}
  A pure fermionic state $\rho_p^{(N)}=|\Psi_p^{(N)}\bracket \Psi_p^{(N)}| \in
  \mathcal{S}(\mathcal{A}_M^{(N)})$ is called a \emph{product state}, if there
  exists a passive transformation $a_k\mapsto a_k'$ such that
\begin{equation}\label{product-state}
|\Psi_p^{(N)}\rangle = \prod_{j=1}^{N}a'^\dag_j |0\rangle.
\end{equation}
A state $\rho_s$ is called \emph{separable}, if it can be written as
the convex combination of product states, i.e.
\begin{equation}\label{separable state}
\rho_s = \sum_{p=1}^K \lambda_p \rho_p^{(N_p)},
\end{equation}
where $\sum_{p=1}^K \lambda_p =1$, $\lambda_p \geq 0$ and all
$\rho_p^{(N_p)} \in \mathcal{S}(\mathcal{A}_M^{(N_p)})$ are
product states. All other states are said to have ''Slater number
larger than 1'' and are called \emph{entangled} (in the sense of
\cite{PhysRevA.64.022303,ESBL-Indistinguishable}).

We denote the set of all separable states by $\mathcal{S}_{sep}$ and by
$\mathcal{S}^{(N)}_{sep} \equiv \mathcal{S}_{sep} \cap
\mathcal{S}(\mathcal{A}_M^{(N)})$ the set of all separable states of particle
number $N$.
\end{definition}
Note that the sets $\mathcal{S}_{sep}, \mathcal{S}_{sep}^{(N)}$ of separable
states are convex and invariant under passive transformations. Both properties
will be useful later on.\\
Separable states have only correlations resulting from
their anti-symmetric nature and classical correlations due to mixing. In the
terminology of Refs.~\cite{PhysRevA.64.022303,ESBL-Indistinguishable} they
have Slater number one and describe unentangled particles. These states will
certainly not contain correlations associated with pairing. (Note that they
can be mode-entangled for an appropriate partition of modes.)

Besides basis change, there are other operations, which do not create quantum
correlations and it is useful to see that the set of separable states is
invariant under them.
%
\begin{lemma}\label{le:sep-trace}
Let $\rho\in \mathcal{S}_{sep}$ be a separable state. Then
the state after measuring the particle number $n_h=\ad_h a_h$ in some mode $h$
is separable for both possible outcomes $n_h=0,1$. \\
Furthermore, $\rho_{h}\equiv \tr_{a_h}[\rho]$, the reduced state
obtained by tracing out the mode  $a_h$, is also separable.
\end{lemma}
\begin{proof}
  As $\mathcal{S}_ {sep}$ is convex, it is sufficient to prove the claim for
  product states $\rho$. Let $|\Psi\rangle = \prod_{j=1}^N \ad_{f_j}\vac $ be
  the vector in Hilbert space corresponding to $\rho$. Our aim is to show that
  $|\Psi\rangle = |\Psi_0\rangle+ |\Psi_1\rangle$, where $|\Psi_l\rangle$ are
  product states and $n_h=l$ eigenstates of the occupation number operator
  $n_h$.  If $h$ is in the span of $\{f_{1\leq k\leq N}\}$ or orthogonal to
  it, the state already is a $n_h$ eigenstate and we are done. Otherwise,
  define $f_{N+1}$ orthogonal to the $f_{k\leq N}$ such that
  $h\in\mathrm{span}\{f_{1\leq k\leq N+1}\}$ and define another orthonormal
  basis $\{g_j\}$ for the span with $g_1=h$ and $g_2\propto f_{N+1}-(h\cdot
  f_{N+1})h$ [here $(h\cdot f_{N+1})$ denotes the inner product on the single
  particle Hilbert space]. Then we can write $|\Psi\rangle =
  a_{f_{N+1}}\ad_{f_{N+1}}\Pi_{j=1}^N\ad_{f_j}\vac =
  (xa_{g_1}+ya_{g_2})\Pi_{j=1}^{N+1}\ad_{g_j}\vac$ for some
  $x,y\in\mathds{C}$.  Hence, $|\Psi\rangle = |\Psi_0\rangle+ |\Psi_1\rangle$
  with $|\Psi_0\rangle = x\Pi_{j=2}^{N+1}\ad_{g_j}\vac$ and $|\Psi_1\rangle =
  -y\ad_{h}\Pi_{j=3}^{N+1}\ad_{g_j}\vac$ which both clearly are product
  states. The reduced state $\tr_h[|\Psi\bracket\Psi|]$ is the statistical
  mixture of $|\Psi_0\rangle$ and $|\Psi_1\rangle$ and therefore clearly
  separable.
\end{proof}

\subsection{Fermionic Gaussian states}\label{Gaussian-section}
Fermionic Gaussian states are represented by density operators that are
exponentials of a quadratic form in the Majorana operators. A general
multi-mode Gaussian state is of the form
\begin{equation}
  \label{eq:Gaussrho}
  \rho= K \exp\left[-\frac{i}{4} c^T G  c\right],
\end{equation}
where $c=(c_1, \ldots, c_{2M})$, $K$ is a normalization constant
and $G$ is a real anti-symmetric $2M \times 2M$ matrix.
Every anti-symmetric matrix can be brought to a block diagonal form
\begin{equation}\label{eq:GaussG}
OGO^T =\bigoplus_{j=1}^M \left(
                        \begin{array}{cc}
                          0 & -\beta_j \\
                          \beta_j & 0 \\
                        \end{array}
                      \right)
\end{equation}
by a special orthogonal matrix $O \in SO(2M)$.

From eq.~(\ref{eq:Gaussrho}) it is clear that Gaussian states have
an interpretation as thermal (Gibbs) states corresponding to a
Hamiltonian $H$ that is a quadratic form in the $c_k$, i.e.,
$H=\frac{i}{4} c^T G c = \frac{i}{4}\sum_{k>l}G_{kl}[c_k,c_l]$ and
the form eq.~(\ref{eq:GaussG}) shows that every Gaussian state has
a normal-mode decomposition in terms of $M$ single-mode ``thermal
states'' of the form $\sim\exp(-\beta \ad a)$. From this one can
see that the state is fully determined by the expectation values
of quadratic operators $a_ia_j$ and $\ad_ia_j$. These are
collected in a convenient form in the real and anti-symmetric
covariance matrix $\Gamma$ which is defined via
\begin{equation}\label{eq:CM}
\Gamma_{kl}=\frac{i}{2}\tr \left(\rho [c_k,c_l] \right).
\end{equation}
It can be brought into block diagonal form by a canonical transformation:
\begin{equation}\label{eq:GammaDiag}
  O \Gamma O^T = \bigoplus_{i=1}^M \left(
    \begin{array}{cc}
      0 & \lambda_j \\
      -\lambda_j & 0 \\
    \end{array}
  \right).
\end{equation}
For every valid density operator, $\lambda_j \in [-1,1]$, and the
eigenvalues of $\Gamma$ are given by $\pm i\lambda_j$. Hence,
every $\Gamma$ corresponding to a physical state has to fulfill
$i\Gamma\leq\one$ or, equivalently, $\Gamma\Gamma^\dag\leq\one$
and to each such $\Gamma$ corresponds a valid Gaussian density
operator where the relation between $G$ and $\Gamma$ is given by
$\lambda_j = \tanh(\beta_j/2)$.  The covariance matrix of the
ground state of $H$ is obtained in the limit $|\beta_j|\to\infty$
i.e., $\lambda_j\to\mathrm{sign}(\beta_j)$. In fact, this shows
that every pure Gaussian state is the ground state to some
quadratic Hamiltonian. The purity of the state can be easily
determined from the covariance matrix as a Gaussian state is pure
if and only if $\Gamma^2 = -\one$
(see, e.g., \cite{Gaussian}).\\
As mentioned, Gaussian states are fully characterized by their covariance
matrix and all higher correlations can be obtained from $\Gamma$ by Wick's
theorem (see, e.g., \cite{Gaussian}) via
\begin{equation}\label{Wick}
i^p \tr[\rho c_{j_1}\ldots c_{j_{2p}}] = \mbox{Pf}(\Gamma_{j_1,
\ldots, j_{2p}}),
\end{equation}
where $1 \leq j_1 < \ldots < j_{2p} \leq 2M$ and $\Gamma_{j_1,
\ldots, j_{2p}}$ is the corresponding $2p \times 2p$ submatrix of
$\Gamma$. $\mbox{Pf}(\Gamma_{j_1, \ldots, j_{2p}})^2= \mbox{det}
(\Gamma_{j_1, \ldots, j_{2p}})$ is called the Pfaffian.

In some cases it is more appropriate to use a different ordering
of the Majorana operators, the so-called q-p-ordering $c=(c_1,
c_3, \ldots, c_{2M-1}; c_2, c_4, \ldots, c_{2M})$, opposed to the
mode-ordering introduced at the beginning. When using the
q-p-ordering, the relation between the real and complex
representation is given by
\begin{equation}\label{trafo-c-a}
c^T = \Omega a^T, \;\;\; \Omega=\left(
                                  \begin{array}{cc}
                                    \mathds{1} & \mathds{1} \\
                                    i \mathds{1} & -i \mathds{1}\\
                                  \end{array}
                                \right),
\end{equation}
where $a = (a_1, \ldots, a_M, \ad_1, \ldots, \ad_M)$. The
transformation matrix $\Omega$ fulfills $\Omega\Omega^{\dagger} =
2\mathds{1}$.\\
In the q-p-ordering the covariance matrix obtains the following
block structure:
\begin{equation}\label{eq:Gammatilde}
\tilde\Gamma = \left(
           \begin{array}{cc}
             \Gamma_q & \Gamma_{qp} \\
             -\Gamma_{qp}^T & \Gamma_p \\
           \end{array}
         \right).
\end{equation}

Finally, for some purposes it is more convenient to use the complex
representation, where the covariance matrix is of the form
\begin{equation}\label{Gammac}
\Gamma_c = \frac{1}{4}\Omega^\dag\tilde\Gamma\bar{\Omega} =
\left(\begin{array}{cc}
Q & R \\
\bar{R} & \bar{Q} \\
\end{array}\right),
\end{equation}
where $Q_{kl} = \langle i/2 [a_k, a_l]\rangle$, $R_{kl}=\langle
i/2 [a_k, \ad_l]\rangle$ and $\bar{Q}$ denotes the complex
conjugate. Note that $R^{\dagger} = -R$ and $Q^T = -Q$ and hence
$\Gamma_c^T=-\Gamma_c$. The condition
$\tilde\Gamma\tilde{\Gamma}^\dag\leq\one$ takes the form
$4\Gamma_c\Gamma_c^\dag\leq\one$.

The description of $\rho$ by its covariance matrix is especially convenient
to describe the effect of canonical transformations, i.e. time evolutions
generated by quadratic Hamiltonians: if $c_k\mapsto \sum_l O_{kl}c_l$ in the
Heisenberg picture then $\Gamma\mapsto O\Gamma O^T$ in the Schr\"odinger
picture. For a passive transformation $a_k \mapsto a_k^{\prime}=
\sum_l U_{kl}a_l$, the q-p-ordered Majorana operators transform as
\begin{equation}
 c^T\mapsto c^{\prime T} = O_p \,c^T, \;\;\;\;  O_p=\left(
\begin{array}{cc}
  X & Y \\
 -Y & X \\
\end{array}
\right),
\end{equation}
where $X = \mbox{Re}(U)$ is the real part of the unitary $U$, and
$Y=\mbox{Im}(U)$ the imaginary part. Note that $O_p$ is both
orthogonal and symplectic. The behaviour of $\Gamma_c$ under a passive
transformation is particularly simple: $Q$ and $R$ transform
according to $Q \mapsto U Q U^T$ and $R \mapsto U R U^{\dagger}$.

Passive transformations can be used to transform \emph{pure}
fermionic states to a simple standard form, the so-called
Bloch-Messiah reduction \cite{BlochMessiah}. The q-p ordered CM
$\tilde\Gamma_{\mathrm{BCS}}$ takes the form (\ref{eq:Gammatilde})
where
\begin{eqnarray}\label{GammaBCS}\label{Gammap}
  \Gamma_q=-\Gamma_p &=& \bigoplus_k \left(
    \begin{array}{cc}
      0 & -2 \mbox{Im}(u_k v_k^*) \\
      2 \mbox{Im}(u_k v_k^*) & 0 \\
    \end{array}
  \right),\\\label{Gammaq}
  \Gamma_{qp} &=& \bigoplus_k \left(
    \begin{array}{cc}
      |u_k|^2-|v_k|^2 & 2 \mbox{Re}(u_k v_k^*) \\
      -2 \mbox{Re}(u_k v_k^*) & |u_k|^2-|v_k|^2 \\
    \end{array}
  \right).
\end{eqnarray}
In Hilbert space, the state in standard form is given by
\begin{equation}\label{BCS-var}
|\Psi_{Gauss}^{(\bn)}\rangle = \prod_k (u_k +v_k
\ad_k\ad_{-k})\vac,
\end{equation}
where $u_k, v_k \in \mathds{C}$, $|u_k|^2+|v_k|^2 = 1$, $\bn =
\sum_k \langle \ad_k a_k\rangle=2\sum_k|v_k|^2$.
This comprises the kind of ''paired'' states appearing in the BCS
theory of superconductivity \cite{PhysRev.108.1175} with $k \equiv
(\vec{k}, \uparrow)$, $-k \equiv (-\vec{k}, \downarrow)$. We will
refer to these states as \emph{Gaussian BCS states}. We would like
to stress the fact that every pure Gaussian state is a Gaussian
BCS state in some basis.

\subsection{Number conserving fermionic states}\label{Connection-N-NN}
For the application to physical systems we are interested in
states for which the particle number is a conserved quantity. We
call $\rho$ a number conserving state if
$[\rho,N_{\mathrm{op}}]=0$ where $N_{\mathrm{op}}$ denotes the
total number operator. Thus, the density operator of a number
conserving state can be written as a mixture of
$N_\mathrm{op}$-eigenstates. In particular, all separable states
as defined in Def.~\ref{DefSep} are number conserving.

The Gaussian BCS wave function \eqref{BCS-var} is not number
conserving (except for the case $\sum_k|u_kv_k|=0$ that either
$u_k$ or $v_k$ vanishes for every mode), but a relation to these
states can be established via the identity
\begin{equation}\label{ncBCS}
|\Psi_{Gauss}^{(\bn)} \rangle = \sum_{N=0}^{2M} \lambda_N
|\Psi_{BCS}^{(N)}\rangle,
\end{equation}
where the number conserving $2N$-particle BCS state is given by
\begin{equation}\label{BCS-N}
|\Psi_{BCS}^{(N)}\rangle = C_N \left(\sum_{k=1}^M \alpha_k
P_k^{\dagger} \right)^N\vac,
\end{equation}
where we have introduced the pair creation operator $\Pd_k = \ad_k
\ad_{-k}$. The coefficients $\alpha_k$ are related to $u_k$ and
$v_k$ via $\alpha_k = v_k/u_k$, $C_N$ is a normalization constant
which is seen to be $$C_N = \left((N!)^2
  \sum_{j_1<\ldots<j_N}|\alpha_{j_1}|^2\ldots|\alpha_{j_N}|^2\right)^{-1/2}$$
by rewriting Eq.~(\ref{BCS-N})  as
\begin{equation}
  \label{eq:BCS-N2}
  C_NN!\sum_{j_1<j_2<\dots<j_N}\alpha_{k_1}\dots\alpha_{k_N}
  \Pd_{k_1}\dots\Pd_{k_N}\vac.
\end{equation}
The coefficients $\lambda_N = \left(\prod_k u_k
\right)/(N! C_N)$ can be interpreted as the probability amplitude
of being in state $|\Psi_{BCS}^{(N)}\rangle$ since $\sum_N
|\lambda_N|^2$=1. We will in general drop the term \emph{number
conserving} and refer to
states of the form \eqref{BCS-N} as \emph{BCS states}. \\
Whenever the distribution of the $\lambda_N$ is sharply peaked
around some average particle number $\bn$, expectation values of
relevant observables for the number conserving BCS states
$|\Psi_{BCS}^{(\bn)}\rangle$ are approximated well by the
expectation values of the Gaussian BCS state. This relation will
turn out very useful later on, as results on Gaussian states can
be translated into results on number conserving BCS states.


%% file: Definition-s1.tex
\section{Pairing theory}\label{Definition-Section}
In this section we introduce a precise definition of pairing as a
property of quantum states.

\subsection{Motivation and statement of the definition}

The simplest system in which we can find pairing consists of two particles and
four modes\footnote{For three modes, all pure two-particle states are of
  product form.}. The prototypical paired state, for example the spin-singlet
of two electrons with opposing momenta, is of the form
\begin{equation}
  \label{eq:pair}
|\Phi\rangle = \frac{1}{\sqrt{2}}\left(\ad_1\ad_2+\ad_3\ad_4\right)\vac.
\end{equation}
The states describing many Cooper pairs in BCS theory are generalizations of
$\ket{\Phi}$.

The state $\ket{\Phi}$ describes correlations between the two particles that
cannot be reproduced by any uncorrelated state and it can be completely
characterized by one- and two-particle expectations (consisting of no more
than two creation and annihilation operators each). This is a characteristic
of the two-particle property ``pairing'' that we propose to make the central
\emph{defining} property of paired states in the general case of many modes,
many particles and mixed states. Since, moreover, we would call the state
$\ket{\Phi}$ paired no matter what basis the mode operators $a_i$ refer to and
we want it to comprise all BCS states, we are led to the following list of
requirements that a sensible definition of pairing should fulfill:
\begin{enumerate}
\item States that have no internal quantum correlation must be unpaired. These
  are the separable states \eqref{separable state}.
\item Pairing must reveal itself by properties related to one-and two-particle
  expectations only.
\item Pairing should be a basis-independent property.
\item The standard ''paired'' states appearing in the description of solid
  state and condensed matter systems, i.e., the BCS states with wave function
  \eqref{BCS-N} must be captured by our definition.
\end{enumerate}
Further, it would be desirable that there exist examples of paired states that
are a resource for some quantum information
application.\\

Let us define:
\begin{definition}
The set of all operators $\{\mathcal{O}_{\alpha}\}_{\alpha}$ on
$\mathcal{A}_M$ which are the product of at most two creation and
two annihilation operators is called the set of \emph{two-particle
operators}. We denote it by $A_2$.
\end{definition}
These operators capture all one- and two-particle properties of a state $\rho$
and should therefore contain all information about pairing. We will call a
state $\rho$ paired, if it can be distinguished from separable states by
looking at observables in $A_2$ alone. This is formalized in the following
definition:
\begin{definition}\label{Pairing-Def}
  A fermionic state $\rho$ is called \emph{paired} if there exists a set of operators
  $\{\mathcal{O}_{\alpha}\}_{\alpha}\subseteq A_2$ such that the expectation
  values $\{\tr[\rho \mathcal{O}_{\alpha}]\}$ cannot be reproduced
  by any separable state $\rho_s\in\mathcal{S}_{sep}$.\\
States that are not paired are called \emph{unpaired}.
\end{definition}
This definition automatically fulfills our first two requirements
by definition. The third, basis independence, clearly holds, since
the set of separable states is invariant under passive
transformations. We will show that the last requirement is met,
both for Gaussian and number conserving BCS states, i.e. all of
them are paired (see Lemma ~\ref{lemma:allBCSpaired} and Subsec.
\ref{Symp-spec}). Moreover, in Sec.\ref{Interferometer-section} we
can show that there exist paired states
that are a resource for quantum metrology.\\

For states with a fixed particle number, i.e. $\rho \in
\mathcal{S}(\mathcal{A}_M^{(N)})$ it is sufficient to compare with
expectation values on $N$-particle separable states $\rho_s^{(N)}
\in \mathcal{S}_{sep}^{(N)}$, as for all other states the
expectation values of $\langle \sum_i n_i \rangle$ and $\langle
\left(\sum_i n_i\right)^2 \rangle$ differ due to the particle
number constraint. For number conserving states, only number
conserving observables lead to non-vanishing expectation values
and one can thus restrict to linear combinations of $\ad_i
a_j$, $\ad_i \ad_j a_k a_l$.\\
For Gaussian states pairing must reveal itself by properties of
the covariance matrix, as all higher correlations can be obtained
from it via Eq.~\eqref{Wick}. This important fact enables us to
give a complete solution of the pairing problem for fermionic
Gaussian states, which we present in
Sec.~\ref{Gaussian-Pairing-section}.

\subsection{Relation of pairing and entanglement}
Paired states are fermionic states exhibiting non-trivial quantum
correlations. In particular, by definition paired states are
inseparable i.e., entangled in the sense of
\cite{PhysRevA.64.022303,ESBL-Indistinguishable}. This raises
immediately the question: Is pairing equivalent to entanglement?
Below, we provide examples of entangled, but unpaired states that
demonstrate that pairing is not equivalent to entanglement (of
particles) but represents a special type of quantum
correlation\footnote{Note that our basis-independent definition
  clearly has no relation to entanglement of modes, which is basis-dependent.
  The product states of Def. \ref{DefSep} can be mode-entangled for some
  choice of partition of modes, e.g.  $1/\sqrt{2}(\ad_1+\ad_2)|0\rangle$ is
  entangled in modes $\ad_1$ and $\ad_2$.}.
\begin{lemma}
There exist states that are entangled according to the Slater rank
concept, but not paired.
\end{lemma}
\begin{proof}
Consider the state $|\Psi_4\rangle =
\frac{1}{2}(\ad_1\ad_2\ad_3\ad_4 + \ad_5 \ad_6
\ad_7\ad_8)|0\rangle$, which is entangled according to the Slater
rank definition. However, one sees immediately that the one-and
two-particle expectations for $|\Psi_4\rangle$ are the same as for
$\rho_s^{(4)} = \frac{1}{2} |\Phi_1\bracket \Phi_1| + \frac{1}{2}
|\Phi_2 \bracket \Phi_2|,$ where $|\Phi_1 \rangle =
\ad_1\ad_2\ad_3\ad_4|0\rangle,$ $|\Phi_2\rangle = \ad_5 \ad_6
\ad_7\ad_8|0\rangle$. Since $\rho_s^{(4)}$ is a product state,
$|\Psi_4\rangle$ is not paired. One can construct further examples
in a similar manner using e.g. other states with higher Slater
rank.
\end{proof}

Since pairing is defined via expectation values of one-and
two-particle operators only, one might wonder whether pairing is related to
entanglement of the two-particle reduced state. To study this relation, we
recall the definition of the \emph{two-particle density operator}
and the closely related \emph{two-particle density
matrix} (see eg. \cite{Coleman}):
\begin{definition}\label{RDM}
Let $\rho$ be the density operator of a fermionic state. Then
$O_{(ij)(kl)}^{(\rho)}=\tr[\rho \ad_i \ad_j a_l a_k]$ is called
the \emph{two-particle reduced density matrix} (RDM). It is usually not
normalized and fulfills $\tr[O^{(\rho)}]= \langle N_{op}^2\rangle -\langle
N_{op}\rangle$.\\
The operator $\rho_2 = O^{(\rho)}/\tr[O^{(\rho)}]$ is called
\emph{reduced two-particle density operator} (RDO).
\end{definition}
Note the crucial difference between the two-particle RDM and the
RDO. While the RDM contains all two-particle correlations of
$\rho$, the RDO corresponds to the two-particle state of any two
particles when the rest of the system is discarded. We would like
to emphasize, that pairing is \emph{not} equivalent to
entanglement of the RDO, and therefore it is a property of the
one- and two-particle expectations:
\begin{lemma}\label{Pairend-not-entangled}
  Let $|\Psi_{BCS}^{(N)}\rangle$ be a number conserving BCS state as defined in
  \eqref{BCS-N} with $\alpha_k = 1\; \forall\; k=1, \ldots, M$. Then its two-particle RDO
  (see Def. \ref{RDM}) $\rho^{(N)}_{BCS,2}$ is always paired. However, $\rho^{(N)}_{BCS,2}$ is
  entangled if and only if $M >3N -2$.
\end{lemma}
The proof is given in Appendix \ref{Ent-app}.

We would like to stress the point that Lemma
\ref{Pairend-not-entangled} shows the existence of paired states
that are not entangled. Having assured that our definition of
pairing does not coincide with entanglement, we now turn to
methods of detecting and quantifying pairing.

%
%

\subsection{Methods for detecting pairing}
Taking Def. \ref{Pairing-Def}, we aim at finding tools that can be
used for the detection and quantification of pairing. These will
be applied to systems of Gaussian states and number conserving
states in Secs. \ref{Gaussian-Pairing-section} and
\ref{Number-conserving-section} respectively.
In this section, we exploit the convexity of the set of unpaired states to
introduce witness operators and obtain a geometrical picture of the set.
The quantification of pairing via pairing measures will be discussed in Sec.
\ref{measure-subsection}.

Given a fermionic density operator, we are interested in an operational method
to determine whether it is paired or not. As in the case of separability, this
simple-sounding question will turn out to be rather difficult to answer in
general.

Starting from Def. \ref{Pairing-Def}, it is clear that the set of unpaired
states is convex. This suggests the use of the Hahn-Banach separation theorem
as a means to certify that a given density operator is not in the set of
paired states.  In analogy to the entanglement witnesses in quantum
information theory \cite{HHH96} we define
\begin{definition}
A \emph{pairing witness} $W$ is a Hermitian operator that fulfills
$\tr[W\rho_u]\geq 0$ for all unpaired states $\rho_u$, and for
which there exists a paired state $\rho$ such that $\tr[W\rho]<0$.
We then say that $W$ \emph{detects} the paired state $\rho$.
\end{definition}
The witness defines a hyperplane in the space of density operators such that
the convex set of unpaired states lies wholly on that side of the plane
characterized by $\tr[\rho W]>0$.  According to the Hahn-Banach theorem
\cite{Rudin91}, for every unpaired state there exists a witness operator which
detects it.
In principle, a witness operator can be an operator involving an
arbitrary number of creation and annihilation operators. However,
since definition of pairing refers only to expectation values of operators in
$A_2$, it is enough to restrict to operators from that set.
This represents a significant simplification both mathematically (witness
operators from a finite dimensional set) and experimentally, since
operators involving more than two-body correlations are typically very
difficult to measure.\\
The construction of entanglement witnesses detecting all entangled
states is an unsolved problem in entanglement theory, and we will
not be able to give a complete solution to the problem of finding
all pairing witnesses either. However, in Section
\ref{Number-conserving-section} we will construct witnesses for a
large subclass of BCS-states by using the correspondence betweenF
number conserving and Gaussian BCS states.


Whether a state $\rho$ is paired or not can be determined from a finite set of
real numbers, namely the expectation values of a hermitian basis
$\{O_\alpha\}$ of $A_2$.  This allows us to reformulate the pairing problem as
a geometric question on convex sets in finite-dimensional Euclidean space,
describe a complete set of pairing witnesses, and deduce a relation to the
ground state energies of quadratic Hamiltonians.

Consider a set $\{O_{\alpha}, \alpha=1,\dots,K\}\subset A_2$ of
hermitian operators in $A_2$ that are not necessarily a basis.
Denote by $\vec{O}$ the vector with components $O_\alpha$. We
define the set of all expectation values of $\vec{O}$ for
separable states
\begin{equation}
C_{\vec{O}}=\left\{\vec{v} = \tr[\vec{O} \rho_s] : \rho_s
  \in S_{sep}\right\}\subset \mathds{R}^K.
\end{equation}
For a state $\rho$ let $\vec{v}_{\rho}\equiv \tr[\vec{O} \rho]$. By
definition, $\rho$ is paired if $\vec{v}_{\rho} \notin
C_{\vec{O}}$. As the set of separable states is convex, so is
$C_{\vec{O}}$. Hence, we can use a result of convex analysis to check if
$\vec{v}_{\rho} \notin C_{\vec{O}}$ (see e.g. \cite{Rockafellar}):

\begin{lemma}\label{v-in-convex}
Let $C \subset \mathds{R}^N$ be a closed convex set, and let
$\vec{v} \in \mathds{R}^N$. Then
\begin{equation}
\vec{v} \in C \Leftrightarrow \; \forall\, \vec{r} \in \mathds{R}^N:
\vec{v}\cdot \vec{r} \geq E(\vec{r})= \inf_{\vec{w}\in
C}\vec{w}\cdot \vec{r}.
\end{equation}
\end{lemma}
For our purposes, this translates in the following result:
\begin{lemma}\label{lemma:completefamily}
  For a vector of observables $\vec{O} = (O_1\dots,O_K)$ let
  $H(\vec{r})=\vec{r}\cdot \vec{O}$ and
  $E(\vec{r}) = \inf_{\rho \in \mathcal{S}_{sep}}\{\tr[H(\vec{r})\rho]\}$.
  Then $W(\vec r)\equiv H(\vec{r})-E(\vec{r})$ is a pairing witness, whenever
  $E(\vec{r})\not=\mathrm{inf}_{\mathrm{all}\,\rho}\{\tr[\rho H(\vec r)]\}$.

  If $\{O_\alpha\}$ form a basis of $A_2$, then $W(\vec{r})$ is a complete
  set of witnesses in the sense that all paired states are detected by some
  $W(\vec r)$, i.e., $\rho$ is unpaired iff $\tr[W(\vec{r})\rho]\geq0
  \,\forall\,
  \vec{r}$.
\end{lemma}
\begin{proof}
  The witness property of $W(\vec{r})$ is obvious from the definition of
  $E(\vec{r})$. \\
For the second part, ``if'' is clear and ``only if'' is seen as
follows: By Lemma~\ref{v-in-convex}, if $\tr[W(\vec{r})\rho]\geq
0\,\forall\,\vec{r}$ then $\vec{v}_\rho\in C$, i.e. the
expectation values can be reproduced by a separable state. But
since all expectation values of operators $\in A_2$ can be
computed from $\vec{v}_\rho$ this implies all two-particle
expectations of $\rho$ can be thus reproduced, i.e. $\rho$ is
unpaired.
\end{proof}

For an $M$-mode system with annihilation operators $a_i$, a standard choice of
$O_\alpha$ is, e.g., given by the real and imaginary parts of
$\{(\ad_{i}\ad_{j}a_{k}a_l)_{i>j,k>l}, (\ad_i\ad_j)_{i> j}, (\ad_ia_j)_{i\geq
  j}\}$, i.e., the dimension of $A_2$ (as a real vector space) is
$K=M^2(M-1)^2/2+2M^2$.

Thus Lemma~\ref{lemma:completefamily} gives a necessary and
sufficient criterion of pairing and provides a geometrical picture
of the pairing problem. While the proof that a state is unpaired
will in general be difficult as it requires knowledge of all
$E(\vec r)$ and experimentally the measurement of a complete set
of observables, practical sufficient conditions for pairing can be
obtained by restricting to a subset $\mathcal{O}\subset A_2$. We
will show in Sec. \ref{geometry-section} that for a certain choice
of $\{O_\alpha\}\subset A_2$ the set $C_{\vec{O}}$ has a very
simple form and allows a good visualization of the geometry of
paired states and the detection of all BCS states up to passive
transformations.

To provide a way to determite $E(\vec{r})$ used in
Lemma~\ref{lemma:completefamily}, we point out an interesting
connection to the covariance matrices $\Gamma_c$ (cf.
Eq.~\eqref{Gammac}) of Gaussian states: even for number conserving
states, $E(\vec{r})$ is given by a quadratic minimization problem
in terms of $\Gamma_c$.

\begin{lemma}
  Let $E(\vec{r})$ and $H(\vec{r})$ be as in Lemma~\ref{lemma:completefamily}
  and let $\vec{O} = \{\ad_i\ad_ja_ka_l,\ad_ia_j\}$ and group the
  components of $\vec{r}$ in two subsets $(\vec{r})_{ijkl}$ and
  $(\vec{r})_{ij}$ corresponding to the one- and two-particle observables,
  respectively.
Then $E(\vec{r})$ is given by a quadratic minimization problem over complex
covariance matrices Eq.~(\ref{Gammac}), in particular the off-diagonal block
$R$ of $\Gamma_c$. We have
\begin{equation}
  \label{eq:EviaR}
  E(\vec{r}) = \inf_{\stackrel{R=-R^\dagger}{4R^2=-\one}}\left\{
    \vec{\gamma}^TM(\vec{r})\vec{\gamma} +
  w(\vec{r})^T\vec{\gamma} \right\},
\end{equation}
where $(\vec{\gamma})_{kl}=\langle \ad_ka_l\rangle =
-iR_{lk}+\frac{1}{2}\delta_{kl}$ and the $\vec{r}$-dependent quantities are
$[M(\vec{r})]_{(ik)(jl)}=-\vec{r}_{ijkl}+\vec{r}_{ijlk}$ and
$[w(\vec{r})]_{kl}=\vec{r}_{kl}$.\\
The minimization can be extended over \emph{all} (not necessarily pure
separable) CMs without changing the result.
\end{lemma}
\begin{proof}
  The minimum $\mathrm{min}_{\rho\in\mathcal{S}_\mathrm{sep}}\{\langle
  H(\vec{r})\rangle_{\rho}\}$ is attained for \emph{pure} separable states,
  i.e., product states. All pure fermionic product states are Gaussian; then
  by Wick's theorem the expectation values of the $O_{ijkl}=\ad_i\ad_ja_ka_l$
  factorize as $\langle \ad_i\ad_ja_ka_l\rangle_{\rho} =
  \langle\ad_i\ad_j\rangle\langle a_ka_l\rangle -
  \langle\ad_ia_k\rangle\langle\ad_ja_k\rangle +
  \langle\ad_ia_l\rangle\langle\ad_ja_k\rangle$. Since product states are also
  number conserving, the first term vanishes. For the other two we use that
  $\langle \ad_ka_l\rangle = -iR_{lk}+\frac{1}{2}\delta_{kl}$, i.e., they only
  depend on the off-diagonal block $R$. The pure state condition
  $\Gamma^2=-\one$ translates into $4R^2=-\one$ for product states $Q=0$.\\
  We could extend over all CMs $\gamma_c$ since only the block $R$ appears in
  the expression to be minimized over and since if $\Gamma_c(Q,R)$ is a valid
  CM then so is $\Gamma_c(0,R)$.
\end{proof}
This lemma provides a systematic way to construct pairing witnesses.
\subsection{Pairing measures}\label{measure-subsection}
It would be desirable if a theory of pairing not only answers the question
whether a state is paired or not, but also quantifies the amount of pairing
inherent in a state. For this purpose, we introduce the notion of a pairing
measure:
\begin{definition}
Let $\rho$ be an $M$-mode fermionic state. A pairing measure is a
map
     $$\mathcal{M} : \rho \mapsto \mathcal{M}(\rho) \in \mathds{R}_+,$$
which is invariant under passive transformations
and fulfills
$\mathcal{M}(\rho)=0$ for every unpaired state $\rho$.
\end{definition}
In addition, it is often useful to normalize $\mathcal{M}$ such that
$\mathcal{M}(\rho_0)=1$ defines the ``unit of pairing''. The pair state
$\ket{\Phi}$ of Eq.~(\ref{eq:pair}) would be an obvious choice for this unit,
but as we see in Sec.~\ref{Gaussian-measure-section} for Gaussian
states a different unit is more natural, therefore we do not include
normalization in the above definition.

In the geometric picture of the previous section, a candidate for
a pairing measure that immediately comes to mind is the distance
of $\vec{v}_\rho$ from the set $C$. This measure is positive, and
it is invariant under passive transformations, as those correspond
to a basis change in the space of expectation vectors.

The computation of this distance is, in general, very difficult
and there is no evident operational meaning to this quantity. In
the following sections we will introduce a different measure that
can be computed for relevant families of states and allow a
physical interpretation in terms of quantifying a resource for
precision measurements.


%% file: Gaussian-Pairing-s1.tex
\section{Pairing for Gaussian states}\label{Gaussian-Pairing-section}
In this section we study pairing of fermionic Gaussian states. We
start with the construction of pairing witnesses in
Sec.~\ref{Gaussian-witness} which will later be a useful guideline
for the construction of pairing witnesses for number conserving
states. Then we derive a simple necessary and sufficient criterion
for pairing of Gaussian states. In Sec.~\ref{angular-section} we
show how pure fermionic Gaussian states can be connected to an
$SU(2)$ angular momentum representation. This picture will guide
us to the construction of a pairing measure.

\subsection{Pairing witnesses for Gaussian states}\label{Gaussian-witness}
Pairing witnesses for pure Gaussian states emerge naturally from the
property that every such state is the ground state of a quadratic
Hamiltonian (see Sec. \ref{Gaussian-section}). This leads to the
following theorem:
\begin{theorem}\label{Gaussian-witness-thm}
Let $0 < \epsilon < 1$ and let $0 \leq |v_k|^2\leq 1-\epsilon$ and
$\sum_k|v_k|^2>0$. Then the operator
\begin{eqnarray}\label{Hquad} H&=&\sum_{k=1}^M
2(1-\epsilon-|v_k|^2)(n_{k }+ n_{-k})\nonumber\\&-& 2 v_k u_k^*\Pd_k
- 2 v_k^* u_k P_k
\end{eqnarray}
is a pairing witness, detecting $$|\Psi_{Gauss}\rangle = \prod_k
(u_k + v_k P^{\dagger}_k)|0\rangle.$$
\end{theorem}

\begin{proof}
Every Gaussian state is the ground state of a quadratic
Hamiltonian. In particular, $|\Psi_{Gauss}\rangle$ is seen to be
the ground state of
\begin{eqnarray*}
H_0&=&\sum_{k=1}^M (|u_k|^2-|v_k|^2)(n_{k }+ n_{-k}-1)\nonumber\\&-&
2 v_k u_k^*\Pd_k - 2 v_k^* u_k P_k
\end{eqnarray*}
with the help of \eqref{GammaBCS}-\eqref{Gammaq}, as the
Hamiltonian matrix of $H_0$ and $\Gamma$ can be brought
simultaneously to the standard forms \eqref{eq:GaussG} resp.
\eqref{eq:GammaDiag}. Subtracting the minimal energy for separable
states, $E_\mathrm{min}^\mathrm{sep}= -(1-2\epsilon)\sum_k \langle
n_k + n_{-k}\rangle -(|u_k|^2-|v_k|^2)$ (note that $\langle
P_k\rangle = 0$ for separable states), we arrive at the
Hamiltonian \eqref{Hquad}. For separable states $\rho$, the
expectation values of $\Pd_k$ vanish, so that $\tr[H \rho]\geq 0$.
For the Gaussian BCS state, however, $\langle \Psi_{Gauss}|
H|\Psi_{Gauss}\rangle = -4 \epsilon \sum_k |v_k|^2 <0$.
\end{proof}

\subsection{Complete solution of the pairing problem for fermionic Gaussian states}
\label{Symp-spec} Every Gaussian state is completely characterized
by its covariance matrix, so that the solution of the pairing
problem must be related to it. The pairing problem is completely
solved by the following theorem:

\begin{theorem}\label{symp-spec-thm}
Let $\rho$ be the density operator of a fermionic Gaussian state
with covariance matrix $\Gamma_c$ defined in \eqref{Gammac}. Then
$\rho$ is paired iff $Q \neq 0$.
\end{theorem}

\begin{proof}
First, note that the condition $Q=0$ is independent of the choice
of basis. If $\rho$ is not paired, then there exists a separable
state having the same covariance matrix as $\rho$. This implies
$Q=0$, as separable states are convex combinations of states with
fixed particle number, and thus $\langle i/2 [a_k, a_l]\rangle = 0$.\\
Now, let $\Gamma_c$ be the covariance matrix of a paired Gaussian
state, and assume that $Q=0$. As $R$ is anti-hermitian, there
exists a passive transformation such that $R_{ij}= r_i
\delta_{ij}$, and $Q=0$ is unchanged. But such a covariance matrix
can be realized by a separable state fulfilling $\langle n_i
\rangle = r_i$ in contradiction to the assumption.
\end{proof}

Note that Thm. \ref{symp-spec-thm} implies that a Gaussian state
is unpaired iff it is number conserving.

\subsection{Angular momentum algebra for Gaussian states}\label{angular-section}
\begin{figure}
\begin{center}
\includegraphics[scale=0.7]{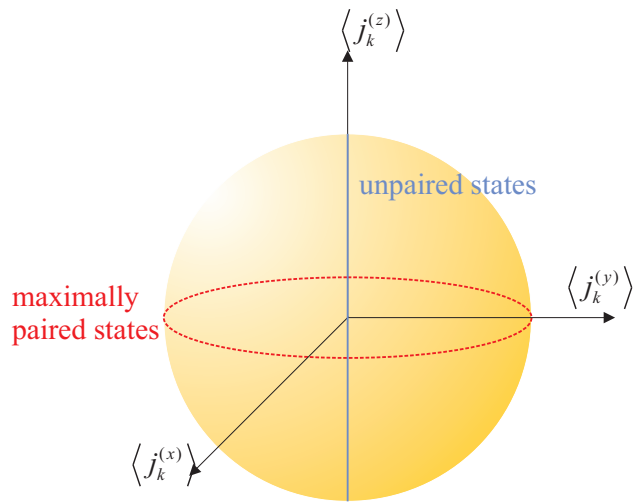}
\caption{Bloch sphere representation of the expectation values of
$j_k^{(x)}$, $j_k^{(y)}$ and $j_k^{(z)}$ for a variational BCS
state. All pure states lie on the surface of the sphere. Unpaired
states lie on the $z$-axis, while the maximally paired states lie
on the equator.\label{Blochsphere}}
\end{center}
\end{figure}
In this section we will show that pairing of Gaussian states can
be understood in terms of an $SU(2)$ angular momentum algebra. The
expectation values of the angular momentum operators can be
visualized using a Bloch sphere, giving us further understanding
of the structure of pairing in Gaussian states. It later leads to
the construction of a pairing measure for these states. Define the
operators \cite{PhysRev.112.1900,PhysRevLett.93.130403}
\begin{eqnarray*}
j_k^{(x)}&=&\frac{1}{2}\left(\Pd_k + P_k \right),\\
j_k^{(y)}&=&\frac{i}{2}\left(\Pd_k - P_k \right),\\
j_k^{(z)}&=&\frac{1}{2}\left(1-n_{k}-n_{-k} \right).
\end{eqnarray*}
They fulfill $\left[j_k^{(a)}, j_k^{(b)}\right]= i
\varepsilon_{abc} j_k^{(c)}$, $a,b,c\in \{x,y,z\}$, forming an
$SU(2)$ angular momentum algebra. For pure Gaussian states in the
standard form \eqref{BCS-var} the expectation values of the
angular momentum operators are given by $\langle j_k^{(x)} \rangle
= \mbox{Re}(u_k\bar{v_k})$, $\langle j_k^{(y)} \rangle =
\mbox{Im}(u_k\bar{v_k})$, and $\langle j_k^{(z)} \rangle
=\frac{1}{2} (1-2 |v_k|^2)$. As  $j^2 = \sum_{i=x,y,z} \langle
j_k^{(i)} \rangle^2= 1/4$ independent of $u_k$ and $v_k$, the
expectation values for every pure Gaussian state lie on the
surface of a sphere with radius $1/2$. As we have shown in Thm.
\ref{symp-spec-thm}, every unpaired state $\rho_u$ fulfills
$\langle j_k^{(x)} \rangle_{\rho_u} = \langle j_k^{(y)}
\rangle_{\rho_u} = 0$, so that these states are located on the
$z$-axis. The states on the equator have $\langle j_k^{(x)}
\rangle^2 +\langle j_k^{(y)}\rangle^2 = 1/4$, i.e. they correspond
to $|u_k|^2=|v_k|^2=1/2$. The situation is depicted in Fig.
\ref{Blochsphere}. Referring to the states on the equator as
maximally paired is suggested by the fact that they have maximal
distance from the set of separable states. This intuitive picture
is further borne out by two observations: first, the states on the
equator display maximal entanglement between the involved modes
\cite{LinearOptics}. Second, they have the
property\footnote{Maximally entangled states of two qubits share
an analogous
  property about entanglement witnesses \cite{bertlmann05}.} that they
achieve the minimal expectation value of any quadratic witness operator (up to
basis change). To see this, recall from Sec.~\ref{Gaussian-section} that any
quadratic Hamiltonian of two modes $k,-k$ is (up to a common factor and
basis change) of the form $\alpha\one+\sin\theta (n_k+n_{-k})+\cos\theta
(P_k^\dagger+P_k)$. It is a witness (i.e., has positive expectation for
all product states), if $\alpha\geq |\mathrm{max}\{0,2\sin\theta\}|$ and does
detect some paired state as long as $\sin\theta>-1$. The minimum eigenvalue is
$\sin\theta-1+\alpha$ and the minimum $\tr(W\rho)=-1$ is attained for $\rho =
\frac{1}{2}(1+P_k^\dag)\proj{0}(1+P_k)$.

The pairing measure which is the topic of the next section will confirm the
characterization as maximally paired.

\subsection{A pairing measure for Gaussian
states}\label{Gaussian-measure-section} The angular momentum
representation of paired states depicted in Fig. \ref{Blochsphere}
suggests the introduction of a pairing measure via a quantity
related to $|\langle j_k^{(x)} \rangle_{\rho_G}|^2 +|\langle
j_k^{(y)} \rangle_{\rho_G}|^2 = |\langle \ad_k
\ad_{-k}\rangle_{\rho_G}|^2$:

\begin{definition}\label{Gaussian-measure}
Let $\rho$ be a fermionic state, and let $Q_{kl} = i/2 \tr(\rho
[a_k, a_l])$. Then we define
\begin{equation}
\mathcal{M}_G(\rho)= 2||Q||^2_2 = 2\sum_{kl}|Q_{kl}|^2,
\end{equation}
\end{definition}

\begin{lemma}
  $\mathcal{M}_G$ as defined in Eq.~(\ref{Gaussian-measure}) is a pairing
  measure fulfilling $\mathcal{M}_G(\rho)\leq M$ for every
  $M$-mode Gaussian state.
\end{lemma}
\begin{proof}
Under a passive transformation $Q \mapsto UQU^T$, and hence
$||Q||_2^2$ is invariant. Further, we know by Thm.
\ref{symp-spec-thm} that $Q=0$ for unpaired states.\\
It remains to show that for an $M$-mode Gaussian state $\rho$ we
have $\mathcal{M}_G(\rho) \leq M$. Let $\Gamma_c$ be the $2M
\times 2M$ covariance matrix of $\rho$ defined in \eqref{Gammac}.
We show first that $\mathcal{M}(\rho)$ is maximized for pure
Gaussian states. To do so, recall that an admissible covariance
matrix for a Gaussian state in the real representation has to
fulfill $i\Gamma  \leq 1$ with equality iff $\Gamma$ is the
covariance matrix of a pure Gaussian state. This translates into
$\Gamma_c \Gamma_c^{\dagger}\leq \mathds{1}$ with equality iff
$\Gamma_c$ belongs to a pure Gaussian state. Using the form of
$\Gamma_c$ given in \eqref{Gammac}, this implies $2 \tr[Q
Q^{\dagger}+RR^{\dagger}]\leq \tr[\mathds{1}]=2M$. Hence,
$||Q||_2^2 \leq M - ||R||_2^2$. It follows that for fixed value of
$||R||_2^2$ the value of $||Q||_2^2$ is maximal for a pure
Gaussian state. Further, the standard form \eqref{BCS-var}implies
that for every value of $||R||$ such a state exists, and that the
maximal value is given by $||Q||_2^2= 2 \sum_{k=1}^M
|u_k|^2|v_k|^2\leq M/2$, as $|u_k|^2+|v_k|^2=1$.
\end{proof}

Hence, for every pure Gaussian state with standard form
\eqref{BCS-var} the value of the pairing measure is given by
$\mathcal{M}_G(\rho) = 4 \sum_{k=1}^M |u_k|^2|v_k|^2$. Since
$|v_k|^2=1-|u_k|^2$ the measure attains its maximum value for
$|u_k|^2=|v_k|^2=1/2$, i.e., for the states already identified as
maximally paired.

$\mathcal{M}_G(\rho)$ will appear again when we study the use of
paired states for metrology applications, linking the pairing measure to the
usefulness of a state for quantum phase estimation and giving support to the
``resource'' character of paired states.


%% file: Number-conserving-s1.tex
\section{Pairing of number conserving
states}\label{Number-conserving-section} In the last section we
gave a complete solution to the pairing problem for fermionic
Gaussian states. There, Wick's theorem lead to a reduction of the
problem to properties of the covariance matrix. For number
conserving systems, the situation is more complicated, as now also
operators of the form $\ad_i \ad_j a_k a_l$ have to be taken into
account. However, we will derive pairing witnesses capable of
detecting all number conserving BCS states in
Sec.~\ref{geometry-section} using the concept of convex sets. For
certain classes of BCS states we will construct a family of
improved witnesses using the analogy to the Gaussian states.
Witnesses have the drawback that they depend on the choice of
basis. I.e. even if a witness detects $\rho$, it does not detect
all states related to $\rho$ by a passive transformation. We will
show that the eigenvalues of the reduced two-particle density
matrix can be used to obtain a sufficient criterion for pairing in
Sec.~\ref{eigenvalue-section} which is basis independent. We close
the section with the construction of a pairing measure in
Sec.~\ref{measure-section}.

\subsection{Pairing of all BCS states and geometry of paired states}\label{geometry-section}
\begin{figure}
\begin{center}
\includegraphics[scale=0.48]{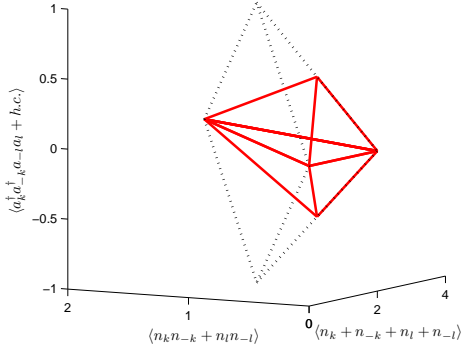}
\caption{(color online) Expectation values of the vector
  Eq.~(\ref{eq:3observables}). For all number conserving states these lie
  within the convex set $C_{\vec{O}_3}^\mathrm{all}$ indicated by the dashed
  back lines. The extreme points of the polytope are given by
  $(0,0,0)$, $(2,0,0)$, $(4,2,0)$ and $(2,1,\pm 1)$. Unpaired states have expectation values in the smaller convex
  set $C_{\vec{O}_3}^\mathrm{unpaired}$ (solid red) which has extreme points
  $(0,0,0)$, $(2,0,0)$, $(4,2,0)$ and $(2,1/2,\pm 1/2)$.
  \label{convexsetfig}}
\end{center}
\end{figure}

In a realistic physical setup it may not be practical to perform
all the measurements needed according to Lemma
\ref{lemma:completefamily} to check the necessary and sufficient
condition for pairing. Having access only to a restricted set of
measurements, necessary criteria for pairing can be derived. In
this section we consider the simplest case of a symmetric
measurement involving four modes, i.e. we are looking at the
following vector of operators:
\begin{equation}
  \label{eq:3observables}
  \vec{O}_3 =\left(
             \begin{array}{c}
               n_k
+n_{-k}+n_{l}+n_{-l} \\
                n_k n_{-k}+n_l n_{-l}\\
               \ad_k\ad_{-k}a_{-l}a_l +h.c. \\
             \end{array}
           \right).
\end{equation}
Remarkably, these expectation values will turn out to be
sufficient to detect all BCS states as paired.\\
We are interested in
$C_{\vec{O}_3}^{\mathrm{unpaired}}=\{\tr(\vec{O}_3\rho) :
\rho\,\,\mathrm{separable}\}$, the set of all expectation values
of $\vec{O}_3$ which correspond to separable states. If for some
$\rho$ the vector $\vec{v}_\rho=\tr(\vec{O}_3\rho)$ is found
outside of $C_{\vec{O}_3}^{\mathrm{unpaired}}$ then it follows
from Lemma~\ref{lemma:completefamily} that $\rho$ is paired.
Membership in $C_{\vec{O}_3}^{\mathrm{unpaired}}$  can be easily
checked by the following Lemma.

\begin{lemma}\label{Convex-Sep}
  A number conserving state $\rho$ has
  expectation values of $\vec{O}_3$ (see Eq.~(\ref{eq:3observables})) compatible with separability if and
  only if $\tr(H^{(p)}_{k\pm}\rho)\geq 0$ for $k=1,2,3$, where
\begin{eqnarray}\label{eq:symwitness} H_{1\pm}^{(p)} &=&
\frac{1}{2}(n_k +n_{-k}+n_{l}+n_{-l})-(n_k n_{-k}+n_l
n_{-l})\nonumber\\
&\pm& (\ad_k\ad_{-k}a_{-l}a_l +h.c.),\\
H_{2\pm}^{(p)} &=& (n_k n_{-k}+n_ln_{-l})\pm (\ad_k\ad_{-k}a_{-l}a_l +h.c.),\label{w2}\\
H_{3\pm}^{(p)} &=& 1-\frac{1}{2}(n_k
+n_{-k}+n_{l}+n_{-l})+\frac{1}{2}(n_k n_{-k}+n_l
n_{-l})\nonumber\\
&\pm& \frac{1}{2}(\ad_k\ad_{-k}a_{-l}a_l +h.c.).\label{w3}
\end{eqnarray}
Hence, the extremal points of the set
$C_{\vec{O}_3}^{\mathrm{unpaired}}$ are given by
$\tr(H^{(p)}_{k\pm}\rho)=0$ for three of the witnesses
~\eqref{eq:symwitness}-\eqref{w3}. The faces of
$C^\mathrm{unpaired}_{\vec{O}_3}$ consist of points for which at
least one of the expectation values $\tr(H^{(p)}_{k\pm}\rho)$ vanishes.\\
$H_{1\pm}^{(p)}$ and $H_{3\pm}^{(p)}$ are also pairing witnesses,
while $H_{2\pm}^{(p)}$ is non-negative on all number conserving
states.
\end{lemma}

\begin{lemma}\label{Convex-all}
  Every number conserving fermionic state fulfills  $\tr(H^{}_{k\pm}\rho)\geq0$, where
\begin{eqnarray}\label{eq:allwitness}
H_1 &=& \frac{1}{2}(n_k +n_{-k}+n_{l}+n_{-l})-(n_k n_{-k}+n_l
n_{-l})\\
H_{2\pm} &=& (n_k n_{-k}+n_ln_{-l})\pm (\ad_k\ad_{-k}a_{-l}a_l +h.c.)\label{wa2}\\
H_{3\pm} &=& 2-\frac{1}{2}(n_k +n_{-k}+n_{l}+n_{-l}) \pm
(\ad_k\ad_{-k}a_{-l}a_l +h.c.).\nonumber\\\label{wa3}
\end{eqnarray}
  The extremal points of the set $C_{\vec{O}_3}^\mathrm{all}=\{\tr(\vec{O}_3\rho) :
  \rho\in\mathcal{S}(\mathcal{A}_M^{(N)}):M,N\in\mathds{N}\}$ are given by
$\tr(H_{k\pm}\rho)=0$ for three of the witnesses
\eqref{eq:allwitness}-\eqref{wa3}. The faces of
$C^\mathrm{all}_{\vec{O}_3}$ consist of points for which at least
one of the expectation values $\tr(H_{k\pm}\rho)$ vanishes.
\end{lemma}

The proofs of the two lemmas can be found in
Appendix~\ref{convex-proof}.  We denote by $C^{\mathrm{unpaired}}$
and $C^\mathrm{all}$ the polytopes containing all expectation
vectors $\vec{v}_\rho$ corresponding to unpaired states or all
number conserving states, respectively. They are bounded by 6
resp. 5 planes defined through the witnesses given in Lemmas
\ref{Convex-Sep} and \ref{Convex-all}. The situation is depicted
in Fig.~\ref{convexsetfig}.

The witnesses $H^{(p)}_{1\pm}$ given in Eqs.~(\ref{eq:symwitness})
allow to detect all number conserving BCS states as paired:
\begin{lemma}\label{lemma:allBCSpaired}
  The number conserving BCS state $\ket{\Psi^{(N)}_{BCS}}$ given in
  Eq.~(\ref{BCS-N}) is (except for the trivially unpaired cases
  $\alpha_k=\delta_{kk_0}$ and $N=M$) detected by the witness $H^{(1)}_p$
  by choosing any two modes $(k,l)$.
\end{lemma}
\begin{proof}
The first two terms in $H^{(p)}_{1\pm}$ are designed such that
their expectation value vanishes for states such as
$\ket{\Psi_{BCS}^{(N)}}$: Since we either have a pair or no
particles in the modes $(k,-k)$ we are in an
eigenstate with eigenvalue 0 of the operators $n_k+n_{-k}-2n_kn_{-k}$.\\
The expectation value of the third term is found using the
representation Eq.~(\ref{eq:BCS-N2}) as
$|C_N|^2N!^2\mathrm{Re}(\alpha_k\alpha_l^*)\sum_{\stackrel{j_1<\ldots<j_{N-1}}{j_i\not=k,l}}|\alpha_{j_1}|^2\ldots|\alpha_{j_{N-1}}|^2$,
which is nonzero unless $N=M$ or all but one $\alpha_k\not=0$. The
sign can be adjusted by a passive transformation to give $\langle
H_{1+}^{(p)}\rangle_{BCS}^{(N)}<0$.
\end{proof}
This shows that indeed all BCS states are paired, as desired.

The witnesses $H^{(1)}_{p\pm}$, while detecting every BCS state as
paired, are in general far from optimal. As the number conserving
BCS states appear in many physical setting, like in the BEC-BCS
crossover \cite{Leggett}, it is desirable to construction improved
witnesses tailored for this class of states. For BCS states
realized in nature it is often appropriate to assume some symmetry
of the wave function $|\Psi^{(N)}_{BCS}(\alpha_k)\rangle =
\left(\sum_{k=1}^{2M} \alpha_k \Pd_k\right)^N|0\rangle.$ For
example, if  $\Pd_k =
\ad_{\vec{k}\uparrow}\ad_{-\vec{k}\downarrow}$, $\Pd_{k+M} =
\ad_{-\vec{k}\uparrow}\ad_{\vec{k}\downarrow}$ and if we are
dealing with an isotropic setting, $\alpha_k =\alpha_{k+M}$ will
hold. It is further often appropriate to assume that the number of
modes is much bigger than the number of particles, i.e. $M \gg N$.
For this kind of states we will construct pairing witnesses via
the correspondence to the Gaussian picture. We sketch the idea of
this construction leading to Thm. \ref{witness-Nc}, and give the
details in the appendix.\\
We have shown in Sec. \ref{Connection-N-NN} the connection of the
Gaussian wave function and the number conserving wave function via
$|\Psi_{Gauss}\rangle = \sum_{k=1}^N \lambda_N
|\Psi^{(N)}_{BCS}(\alpha_k)\rangle$. Consider a number conserving
observable $O$ and denote by $\langle O\rangle_{Gauss}$ and
$\langle O\rangle_N$ its expectation value for the Gaussian and
$2N$-particle BCS wave function respectively. If the distribution
of $|\lambda_N|^2$ is sharply peaked around some average particle
number $\bn$ with width $\Delta$, then  $\langle O\rangle _{Gauss}
\approx \langle O\rangle_N$ for any integer $N \in [\bn -\Delta,
\bn + \Delta]$. In Thm. \ref{Gaussian-witness-thm} we have
constructed witnesses $H$ for all Gaussian BCS states. As these
witnesses are optimal, they suggest to constitute an improved
witness detecting the corresponding number conserving BCS state.
But $H$ includes terms of the form $\Pd_k$ that do not conserve
the particle number. Hence, this witness cannot be applied
directly to the number conserving case. Using Wick's theorem,
$\langle \Pd_k P_{k+M}\rangle_{Gauss} = \bar{u}_kv_k\langle
\Pd_k\rangle_{Gauss}$ holds under our symmetry assumption. This
suggests that we replace the non-number conserving operator
$\bar{u_k}v_k \Pd_k$ by the number conserving operator $\Pd_k
P_{k+M}$. We define operators
\begin{eqnarray}
H_k &=&2(1-\epsilon
-|v_k|^2)N_k- 4 (\Pd_k P_{k+M} + h.c.),\label{Hk}\\
N_k &=& n_k + n_{-k}+n_{k+M} + n_{-(k+M)},
\end{eqnarray}
where $0\leq |v_k|^2\leq 1-\epsilon\;\forall k$ for $\epsilon
>0$. Further, we introduce the notation $\alpha_k =
v_k/\sqrt{1-|v_k|^2}$, $\bn= \sum_{k=1}^M|v_k|^2$ and we denote by
$N$ the biggest integer fulfilling $\bn -\tilde{N} \geq 0$. Then
the following holds:

\begin{theorem}\label{witness-Nc}
Let $M, \tilde{N} \in \mathds{N}$ and let $1 \ll N < 2M$. If $1 >
\epsilon \geq 18/\sqrt{\pi \bn}$ the Hamiltonian
$H(\{v_k\})=\sum_{k=1}^M H_k$ is a pairing witness detecting
$$|\Psi_{BCS,sym}^{(N)}\rangle =
C_N \left(\sum_{k=1}^M \alpha_k(\Pd_k + \Pd_
{k+M})\right)^N\vac.$$
\end{theorem}
The proof is given in Appendix \ref{Proof-witness}.

\subsection{Eigenvalues of the two-particle reduced density
matrix}\label{eigenvalue-section} In this section we derive a
basis independent condition for detecting pairing. The
two-particle reduced density matrix $O$ contains all two-particle
correlations. As a change of basis, $\ad_i \mapsto \sum_k
U_{ik}\ad_k$ leaves the spectrum of $O$ unchanged since
$$O^{(\rho)}_{(ij),(kl)}\rightarrow (U\otimes
U)_{(ij),(mn)}O^{(\rho)}_{(mn),(pq)}(U\otimes
U)^{\dagger}_{(pq),(kl)},$$ we are lead to the following theorem:

\begin{theorem}\label{Eigenvalue-Thm}
Let $\rho$ be an unpaired state, and let $O$ be its two-particle
RDM $O$. Then $\lambda_{max}(O)\leq 2$, where $\lambda_{max}$
denotes the maximal eigenvalue.
\end{theorem}

\begin{proof}
If $\rho$ is unpaired, then there exists a separable state $\rho_s
\in \mathcal{S}_{sep}$ having the same two-particle RDM.  Any
separable state is of the form $\rho_s = \sum_{\alpha}
\mu_{\alpha} \rho^{(\alpha)}$, where $\rho^{(\alpha)} =
|\psi^{(\alpha)}\bracket \psi^{(\alpha)} |, \; |\psi^{(\alpha)}
\rangle = \prod_i \ad_{\alpha_i}\vac$ and $\sum_{\alpha}
\mu_{\alpha} = 1$. Here, $\{\ad_{\alpha_i}\}_i$ denotes some basis
of mode operators. The RDM is of the form $O^{(\rho)} =
\sum_{\alpha} \mu_{\alpha} O^{({\alpha})}$, where $O^{({\alpha})}$
is the RDM for the state $\rho^{(\alpha)}$. The RDM is calculated
in the basis $\{\ad_i\}_i$, and the different bases are related by
a unitary transformation $\ad_i = \sum_j
U^{(\alpha)}_{ik}\ad_{\alpha_k}$, so that
$O^{(\alpha)}_{(ij)(kl)}= \tr[\rho^{(\alpha)} \ad_i \ad_j
a_la_k]=(U^{(\alpha)}\otimes
U^{(\alpha)})_{(ij)(mn)}O_{(mn),(pq)}^{(\alpha,0)}(U^{(\alpha)}\otimes
U^{(\alpha)})^{\dagger}_{(pq)(kl)}$, where
$O_{(mn),(pq)}^{(\alpha,0)}=\langle\ad_{\alpha_m}\ad_{\alpha_n}a_{\alpha_q}a_{\alpha_p}\rangle_{\rho^{(\alpha)}}$.
In the basis of the $\{\ad_{\alpha_i}\}_i$ the expectation value
$\langle\ad_{\alpha_m}\ad_{\alpha_n}a_{\alpha_q}a_{\alpha_p}\rangle_{\rho^{(\alpha)}}$
is of the simple form
$\langle\ad_{\alpha_i}\ad_{\alpha_j}a_{\alpha_l}a_{
\alpha_k}\rangle_{\rho^{(\alpha)}} =
\delta_{ik}\delta_{jl}-\delta_{il}\delta_{jk}$. Hence, the
spectrum of the $O^{(\alpha)}$ is given by
$\spec(O^{(\alpha)})=\{0,2\}\; \forall\, \alpha.$ The two-particle
RDM is hermitian as $O^{\dagger}_{(ij)(kl)}= \bar{O}_{(kl)(ij)}=
\overline{\langle \ad_k \ad_l a_j a_i \rangle}= \langle \ad_i
\ad_j a_l a_k \rangle= O_{(ij)(kl)}.$ Then Weyl's theorem
\cite{H&J} implies $\lambda_{max}\left(\sum_{\alpha}
\mu_{\alpha}O^{\alpha}\right)\leq
\sum_{\alpha}\mu_{\alpha}\lambda_{max}(O^{\alpha})\leq
\sum_{\alpha}2\mu_{\alpha} \leq 2$.
\end{proof}

An example of a state detected as paired via criterion is the
BCS-state \eqref{BCS-N} with $N=2, M=3$ and all $\alpha_k$ equal.
The largest eigenvalue of its two-particle RDM is given by
$\lambda_{max}= 8/3$.

\subsection{A pairing measure for number conserving
states}\label{measure-section} In Sec.
\ref{Gaussian-measure-section} we have derived a pairing measure
for Gaussian states. The correspondence with number conserving BCS
states will be a guideline to derive a measure for number
conserving states. However, the measure of
Def.~\ref{Gaussian-measure} involves expectation values of the
form $\langle \ad_k \ad_{-k}\rangle$ that vanish for states with
fixed particle number. Yet, Wick's theorem suggests that a
quantity involving expectation values of the form $\langle \Pd_k
P_l\rangle$ will lead to a pairing measure. This is indeed the
case, which is the content of the following theorem:

\begin{theorem}\label{measure-bound-separable-state}
Let $\rho$ be a number conserving pure fermionic state. Then the
following quantity defines a pairing measure:
\begin{eqnarray}
&&\lefteqn{
\mathcal{M}(\rho)=}\nonumber\\
&&\max\left\{\max_{\{\ad_i\}_i}\sum_{kl=1}^M|\langle
\Pd_kP_l\rangle_{\rho}|  - \frac{1}{2}\sum_k \langle n_k
\rangle_{\rho}
,0\right\},\nonumber\\
\end{eqnarray}
where $\Pd_k = \ad_k \ad_{-k}$ and the maximum is taken over all
possible bases of modes $\{\ad_i\}_i$. For mixed states $\rho$, a
measure can be defined via
\begin{equation}
\mathcal{M}(\rho)=\min \sum_i p_i \mathcal{M}(\rho_i),
\end{equation}
where the minimum is taken over all possible decompositions of $\rho
= \sum_i p_i \rho_i$ into pure states $\rho_i$.
\end{theorem}

\begin{proof}
The positivity of $\mathcal{M}$ and its invariance under passive
transformations follow directly from the definition. It remains to
show that $\mathcal{M}$ is zero for separable states. We will
prove in Lemma \ref{M-on-product} Appendix \ref{App-bound-on-prod}
that any separable state of $2N$ particles fulfills
$\sum_{kl}|\langle \Pd_kP_l\rangle| \leq N$, and that this bound
can always be achieved, which concludes the proof.
\end{proof}

We close the section by calculating the value of the pairing
measure for two easy examples. Let

\begin{eqnarray}
|\Psi_s\rangle &=& \bigotimes_{k=1}^N\frac{1}{\sqrt{2}}\left(\Pd_k
+ \Pd_{-k}\right)\vac,\\
|\Psi_{BCS}^{(N,M)} \rangle &=&C_N
\left(\sum_{k=1}^M\Pd_k\right)^N\vac,
\end{eqnarray}
the tensor product of $N$ spin-singlet states and the BCS state
with equal weights, respectively. These states have a pairing
measure $\mathcal{M}(|\Psi_s\rangle)= N$ resp.
$\mathcal{M}(|\Psi_{BCS}^{(N,M)} \rangle=N(M-N)$. Thus, for the
spin singlet the pairing measure has in addition the property that
it is normalized to 1 and additive, while it is subadditive for
$|\Psi_{BCS}^{(N,M)} \rangle$. Further, this example suggests that
the pairing of $\mathcal{M}(|\Psi_{BCS}^{(N,M)} \rangle=N(M-N)$ is
stronger than for $|\Psi_s\rangle$. We will see indeed in
Subsec.~\ref{Sec:pairinteraction} that states of the form $|\Psi_e
\rangle$ allow interferometry at the Heisenberg limit.


%% file: Interferometer-s1.tex
\section{Interferometry}\label{Interferometer-section}
The goal of quantum phase estimation is to determine an unknown
parameter $\varphi$ of a Hamiltonian $H_{\varphi}=\varphi H$ at
highest possible accuracy. The value of $\varphi$ is inferred by
measuring an observable $O$ on a known input state that has
evolved under $H_{\varphi}$. In a region where the expectation
value $\langle O(\varphi)\rangle$ is bijective, $\varphi$ can be
inferred by inverting $\langle O(\varphi)\rangle$. In a realistic
setup, however, $\langle O(\varphi)\rangle$ cannot be determined,
as this would require an infinite number of measurements. Instead,
one uses the mean value of the measurement results, $o$, as an
estimate of $\langle O(\varphi)\rangle$. This will result in an
error $\delta \varphi$ for the parameter to be estimated, as for a
given value of $\varphi$ we have $\langle O(\varphi)\rangle = o
\pm \sqrt{\mathrm{Var}(o)}$. Linearizing around the real value of
$\varphi$, it follows that the uncertainty of $\varphi$ is given
by \cite{PhysRevA.50.67, PhysRevA.54.R4649}
\begin{equation}\label{varest}
\langle (\delta \varphi)^2
\rangle=\frac{\mathrm{Var}(O)}{|\partial\langle O\rangle/\partial
\varphi|^2},
\end{equation}
where $\mathrm{Var}(O)=\langle O^2\rangle - \langle O\rangle^2$,
and we have used the fact that $\mathrm{Var}(O) = \mathrm{Var}(o)
$. Further, it can be shown that the minimal uncertainty of
$\varphi$ is bounded by \cite{PhysRevLett.72.3439,BCM-long}
\begin{equation}\label{estviavar}
\langle (\delta \varphi)^2 \rangle \mathrm{Var}(H) \geq \frac{1}{4
\nu},
\end{equation}
where $\nu$ is the number times the estimation is repeated. Eq.
\eqref{estviavar} derives from the Cram\'{e}r-Rao bound and is
asymptotically achievable in the limit of large $\nu$.\\
For a given measurement scheme, i.e. for a given input state and a
given observable $O$, the uncertainty in $\varphi$ can be reduced
by using $N$ identical input states and average over the $N$
measurement outcomes. As the preparation of a quantum state is
costly, a precision gain which has a strong dependence on $N$ is
highly desirable. If these probe states are independent of each
other, the precision scales like $1/\sqrt{N}$. This is the
so-called standard quantum limit (SQL). Using distinguishable or
bosonic systems, this limit can be beaten by a factor of
$\sqrt{N}$ by using number-squeezed input states
\cite{eckert:013814, PhysRevA.33.4033,
PhysRevLett.71.1355,PhysRevA.56.R1083}, $N$-particle NOON states
or maximally entangled GHZ-states $\frac{1}{\sqrt{2}}(|N,0\rangle
+ |0,N\rangle)$ \cite{ PhysRevA.54.R4649, PhysRevA.61.043811,
PhysRevA.66.023819, PhysRevLett.79.3865}. Achieving this so-called
Heisenberg-limit is
the big goal of quantum metrology.\\
Less is known for fermionic states where number squeezing and
coherent $N$-particle states are prohibited by statistics.
Nevertheless, there exist fermionic $N$-particle state which can
achieve the Heisenberg limit for phase measurements in a
Mach-Zehnder interferometer setup \cite{PhysRevLett.56.1515}.
Taking the existence of such states as a starting point, we show
that paired fermionic states can be used as a resource for phase
estimation beyond the SQL. We will consider two different
settings. The first setting will be the standard
Ramsey-interferometer setup of metrology, where the coupling
Hamiltonian is proportional to the number operator. Here, we will
see that paired states lead to a precision gain of a factor of 2
compared to separable states. The second setup involves a more
complex coupling. Here it will turn out that by using paired
states the Heisenberg limit, i.e. a phase sensitivity $(\delta
\varphi)^2 \sim 1/N^2$, can be achieved.

\subsection{Ramsey interferometry with fermions}
\subsubsection{General setup}
\begin{figure}
\begin{center}
\includegraphics[scale = 0.8]{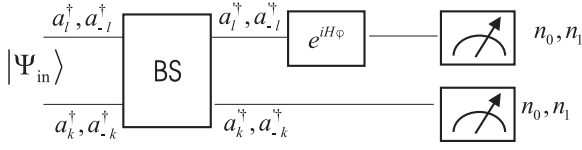}
\caption{Scheme of the Ramsey interferometer setup. The incoming
wave function $|\Psi_{in}\rangle$ enters a beam splitter (BS).
Then particles in the modes $\ad_{\pm l}$ evolve under the
Hamiltonian $\varphi H$. At the end a particle number measurement
is performed on all particles. \label{SI}}
\end{center}
\end{figure}

We consider the standard Ramsey interferometer setup where a state
in the modes $\{\ad_{k_j},\ad_{l_j}\}_{j=-M}^M$ undergoes mode
mixing at a beam splitter,

\begin{eqnarray}\label{Beamsplitter}
\ad_{\pm k_j} &\rightarrow& a^{\prime
\dagger}_{k_j}=\frac{1}{\sqrt{2}}(\ad_{\pm k_j}+\ad_{\pm l_j}),
\\
\ad_{\pm l_j} &\rightarrow& a^{\prime \dagger}_{\pm
l_j}=\frac{1}{\sqrt{2}}(\ad_{\pm k_j}-\ad_{\pm l_j}),
\end{eqnarray}
before evolving under the action of the Hamiltonian
\begin{equation}\label{Hphi}
H_N = \sum_{j=1}^M (n_{l_j}+n_{-l_j}).
\end{equation}
Finally, a particle number measurement is performed on the system,
to compute the parity
\begin{equation}\label{parity}
\mathcal{P} = (-1)^{\sum_{j}n_0^{(j)}+n_1^{(j)}},
\end{equation}
where $n_0^{(j)}=a^{\prime \dagger}_{k_j}a^{\prime}_{k_j}$ and
$n_1^{(j)}=a^{\prime \dagger}_{-k_j}a^{\prime}_{-k_j}$. According to
eq. \eqref{varest}, the phase sensitivity is given by
\begin{equation}\label{deltaphi}
(\delta \varphi)^2 = \frac{1 - \langle \mathcal{P}
\rangle^2}{\left|\frac{\partial}{\partial \varphi}\langle
\mathcal{P} \rangle \right|^2},
\end{equation}
where we have exploited $\mathcal{P}^2 = 1.$ Due to the fermionic
statistics the parity operator can be written in the form
\begin{equation}
\mathcal{P} = \prod_{j=1}^M
\left(1-2(n_0^{(j)}+n_1^{(j)})+4n_0^{(j)}n_1^{(j)}\right).
\end{equation}
In the next section we will derive the best possible precision
obtainable by using unpaired states, and compare this result to
the precision achievable by using paired states. It will turn out
that already at two-particle level paired states have more power
than the unpaired states for our setup.

\subsubsection{Bound on unpaired states for the standard interferometer}
In this section we derive a lower bound on the phase sensitivity
when using an unpaired state of $2N$ particles as input states.

\begin{theorem}
For the Ramsey interferometer described above the phase sensitivity
is bounded by
\begin{equation}\label{phiminprod}
(\delta \varphi)^2\geq \frac{1}{2\nu N},
\end{equation}
when an unpaired state of $2N$ particles is used as input state.
\end{theorem}

\begin{proof}
We will use \eqref{estviavar} to derive the bound. Hence, we have
to estimate an upper bound for the variance of the Hamiltonian
$H_N$ defined in \eqref{Hphi}. As $H_N$ as well as $H_N^2$ contain
operators from the set $A_2$ only, it is sufficient to proof the
bound for product states, as for every unpaired state there exists
a product state having the same expectations. In Lemma
\ref{exp-as-proj} of Appendix \ref{usefull} we have shown that for
pure separable states $\langle n_k n_l\rangle =
|P_{kl}|^2-P_{kk}P_{ll}+P_{kk}\delta_{kl}$, where $P\in
\mathds{C}^{4M \times 4M}$ is a projector of rank $2N$. We arrange
the indices as $-l_M, \ldots, l_M, -k_M, \ldots, k_M$ and
partition the projector $P$ such that $P =\left(
                                                               \begin{array}{cc}
                                                                 A & B \\
                                                                 B^{\dagger} & C \\
                                                               \end{array}
                                                             \right),
$ where $A, B, C \in \mathds{C}^{2M \times 2M}$. Then
\begin{equation}
\mathrm{Var}(H_N)= \sum_{i=1}^{2M} A_{ii}-\sum_{i,
j=1}^{2M}|A_{ij}|^2 = \tr[BB^{\dagger}],
\end{equation}
as $H_N$ only involves the modes $-l_M\ldots, l_M$. In the last
step we have used $P^2=P$, implying $A-A^2=BB^{\dagger}$. As
$\mbox{rank}(P) =2N$, there exists some unitary $U$ such that $P =
U \mathrm{Id}_{2N}U^{\dagger}$, where $\mathrm{Id}_{2N}=\left(
                                                        \begin{array}{cc}
                                                          \mathds{1}_{2N} & 0 \\
                                                          0 & 0 \\
                                                        \end{array}
                                                      \right) \in
                                                      \mathds{C}^{4M\times 4M}
$. Partitioning the unitary $U=\left(
                                 \begin{array}{cc}
                                   U_{11} & U_{12} \\
                                   U_{21} & U_{22} \\
                                 \end{array}
                               \right)
$, where $U_{ij} \in \mathds{C}^{2M \times 2M}$, $i, j = 1, 2$, the
projector $P$ is of the form
$$P = \left(
         \begin{array}{cc}
           U_{11}\mathrm{Id}_{2N}U_{11}^{\dagger} & U_{11}\mathrm{Id}_{2N}U_{21}^{\dagger} \\
           U_{21}\mathrm{Id}_{2N}U_{11}^{\dagger} & U_{21}\mathrm{Id}_{2N}U_{21}^{\dagger} \\
         \end{array}
       \right).
$$
Using the the above representation of $P$ and the cyclicity of the
of the trace, we can write
$\mathrm{Var}(H_N)=\tr[\tilde{A}\tilde{B}]$ with hermitian
matrices
$\tilde{A}=\mathrm{Id}_{2N}U_{11}^{\dagger}U_{11}\mathrm{Id}_{2N}$,
$\tilde{B}=
\mathrm{Id}_{2N}U_{21}^{\dagger}U_{21}\mathrm{Id}_{2N}$. The trace
can be interpreted as a scalar product maximized for linearly
dependent $\tilde{A}$ and $\tilde{B}$. Exploiting the unitarity of
$U$, one sees immediately that the variance is maximized for
$\tilde{A} =c/(1+c) \mathrm{Id}_{2N}$, $\tilde{B} =
1/(1+c)\mathrm{Id}_{2N}$ for some constant $c$. Hence,
$\mathrm{Var}(H_N)\leq c/(1+c)^2\tr[\mathds{1}_{2N}]\leq N/2.$
Inserting this into eq. \eqref{estviavar}, we find that $(\delta
\varphi)^2 \geq \frac{1}{2 \nu N}$.
\end{proof}

\subsubsection{Interferometry with two particles}
In this section we will show that already a two-particle paired
state can beat the bound for the phase sensitivity using unpaired
states \eqref{phiminprod}. Hence pairing manifests itself as
useful quantum correlation already at the two-particle level. We
show the following:

\begin{theorem}\label{Inf2}
Using the paired state
\begin{equation}
|\Psi_{in}^{(2)}\rangle = \left(\sum_{j=1}^M \alpha_j
\ad_{k_j}\ad_{-k_j} + \beta_j \ad_{l_j}\ad_{-l_j}\right)|0\rangle,
\end{equation}
with normalization $\sum_j |\alpha_j|^2+|\beta_j|^2 =1$ as input
state for the Ramsey interferometer, the optimal phase sensitivity
is given by
\begin{equation}\label{phipaired2}
(\delta
\varphi)^2_\mathrm{min}=\frac{1}{2\left(1+2\sum_{j=1}^M\mathrm{Re}(\alpha_j\bar{\beta}_j)\right)}\geq
\frac{1}{4}.
\end{equation}
\end{theorem}

\begin{proof}
Take $|\Psi_{in}^{(2)}\rangle$ as the input state. After an
application of the beam splitter transformation \eqref{Beamsplitter}
and an evolution under the Hamiltonian  \eqref{Hphi}, the
measurement outcome of the parity operator is calculated to be
\begin{equation}
\langle \mathcal{P} \rangle = 1-\sin^2\varphi \left(1+2
\sum_{j=1}^M \mbox{Re}(\alpha_j \bar{\beta}_j)\right).
\end{equation}
Using eq. \eqref{deltaphi}, we obtain \eqref{phipaired2}. The
bound of $1/4$ can be obtained for a state where $\alpha_k =
\beta_k \; \forall\; k.$
\end{proof}

Thm. \ref{Inf2} shows that there exist two-particle paired states
exceeding the bound on product states \eqref{phiminprod}.

\subsubsection{Interferometry with $2N$-particle BCS states}
Generalizing the result obtained in the last section, it follows
immediately that states of the form
$|\Psi_{in}^{(2)}\rangle^{\otimes N}$, will lead to a phase
sensitivity $(\delta
\varphi)^2_\mathrm{min}=1/\left[2N(1+2\sum_{j=1}^M\mathrm{Re}(\alpha_j\bar{\beta}_j))\right]$.
In this section we will show that the same result can be achieved
using BCS states.

\begin{theorem}
Let the paired state
\begin{equation}
|\Psi_{in}^{(2N)}\rangle =c^{\prime} \left(\sum_{j=1}^M \alpha_j
\ad_{k_j}\ad_{-k_j} + \beta_j
\ad_{l_j}\ad_{-l_j}\right)^N|0\rangle,
\end{equation}
where we use the normalization condition $\sum_j
|\alpha_j|^2+|\beta_j|^2 =1$ be the input state for the Ramsey
type interferometer defined above. Then the optimal phase
sensitivity is given by
\begin{equation}\label{phiminBCS}
(\delta \varphi)^2 = \frac{1}{2\bar{N} (1+ 2\sum_j
\mathrm{Re}(\alpha_j \bar{\beta}_j))}.
\end{equation}
\end{theorem}

\begin{proof}
Like in previous sections we will use the correspondence to the
Gaussian state
$$|\Psi_{in,Gauss}^{(2\bar{N})}\rangle=c \exp\left[\sum_{j=1}^M \alpha_j
\ad_{k_j}\ad_{-k_j} + \beta_j \ad_{l_j}\ad_{-l_j}\right]\vac,$$
where $|N-\bn|\ll \bn$ for the calculation. After the state has
passed through the interferometer, the expectation value of the
parity operator is readily computed to be $\langle \mathcal{P}
\rangle_{Gauss} = \prod_j (1-|c|^2|\alpha_j +\beta_j|^2 \sin^2
\varphi)$. As only number operators are involved, $\langle
\mathcal{P} \rangle_{Gauss} \approx \langle \mathcal{P}
\rangle_{N},$ where $\langle\ldots\rangle_N$ denotes the
expectation value of $\mathcal{P}$ for the state
$|\Psi_{in}^{(2N)}\rangle$. Expanding \eqref{deltaphi} for small
value of $\varphi$ and using $\bn =
|c|^2\sum_k|\alpha_k|^2+|\beta_k|^2 = |c|^2$, one obtains
\eqref{phiminBCS}, which has minimal value $(\delta \varphi)^2 =
1/(4N)$. This result is obtained when $\alpha_j = \beta_j \;
\forall \,j.$
\end{proof}

The above result shows that paired states result in a precision gain
of up to a factor of $2$ compared to the best precision obtainable
for unpaired states \eqref{phiminprod}.\\
The pairing measure derived in Sec. \ref{Gaussian-measure-section}
and \ref{measure-section} quantifies the precision gain obtainable
by the use of paired states. To see this, denote by $|\Psi_{in,
Gauss}^{(2\bar{N})\prime}\rangle$ the state after the beam
splitter transformation. Then the pairing measure
\eqref{Gaussian-measure} for this state evaluates to
$$\mathcal{M}_G(|\Psi_{in,Gauss}^{(2\bar{N})\prime}\rangle)=\frac{N^2}{2}\left(1+2\sum_j
\mbox{Re}(\alpha_j\bar{\beta}_j)\right),$$
so that
\begin{equation}
(\delta \varphi)^2 =
\frac{\bar{N}}{4\mathcal{M}_G(|\Psi_{in,Gauss}^{(2\bar{N})\prime}\rangle}.
\end{equation}
The above relation demonstrates that $\mathcal{M}$ is indeed quantifying a
useful resource present in paired states. Whether this interpretation can be
extended to mixed states will not be explored here.

\subsection{Interferometry involving a pair-interaction
Hamiltonian}\label{Sec:pairinteraction}
\begin{figure}
\begin{center}
\includegraphics[scale = 0.8]{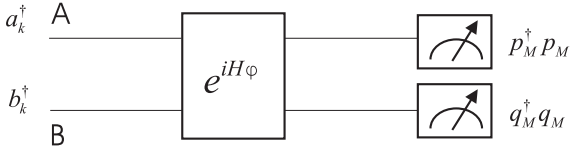}
\caption{Setup which allows interferometry with paired states at the
Heisenberg limit. Particles in modes $\ad_k$ and $b^{\dagger}_k$
evolve under the complex coupling Hamiltonian $H$ (for the detailed
form of $H$ refer to the text). In the end particle numbers are
measured.\label{complex-setup}}
\end{center}
\end{figure}
So far we have seen that paired states lead to a gain of a factor of
2 in precision compared to unpaired states in a Ramsey-type
interferometer. This section will show that paired states are even
more powerful and can lead to a precision gain of a factor of $N$
when measuring the phase of a pair-interaction Hamiltonian.\\
We consider a setup where two fermionic states enter the ports A
and B of an interferometer. The particles entering port A can
occupy the modes $\{\ad_k\}_{k=-M}^M$, while the particles
entering through port B can occupy the modes
$\{b_k^{\dagger}\}_{k=-M}^M$. Then the two states evolve under the
Hamiltonian $H_c$ to be defined below and a particle number
measurement is performed at the end. The situation is depicted in
Fig \ref{complex-setup}. We will compare the power of paired
states over unpaired ones for two different settings. We start by
introducing some basic notation:

\subsubsection{Prerequisites}
We define pair operators $\Pd_k = \ad_k \ad_{-k}$ and
$\Qd_k=b^{\dagger}_k b^{\dagger}_{-k}$ and their equally weighted
superpositions
\begin{equation}
\pdm=\frac{1}{\sqrt{M}}\sum_{k=1}^M\Pd_k,\hspace{0.5cm}
\qdm=\frac{1}{\sqrt{M}}\sum_{k=1}^M\Qd_k.
\end{equation}
The operators $\pdm$ and $\qdm$ fulfill the commutation relations
\begin{eqnarray}\label{commrel}
[\pdm, \ppm]&=& -1
+\frac{1}{M}\Na \label{comma}\\
\mbox{[}\qdm, \qm]&=&-1 +\frac{1}{M}\Nb,\label{commb}
\end{eqnarray}
where $n_k = \ad_k a_k$ so that $N_a=\sum_k (n_k + n_{-k})$, and
$N_b=\sum_k (m_k + m_{-k})$ with $n_k = \ad_k a_k$ and $m_k =
b_k^{\dagger}b_k$ being the number operators for particles in
modes $\ad_k$ and $b_k^{\dagger}$
respectively.\\
We will compare the power of two paired states and two unpaired
states entering through port A and B. The bound for unpaired
states will be derived again via \eqref{estviavar}. As $H_c$ and
$H_c^2$ will be elements of $A_2$, it is like in the last section
sufficient to compare the power of paired states to those of
sparable states. The paired states will be of the form
\begin{eqnarray}\label{paired-in}
|\Psi^{(M)}_N\rangle &=& |N\rangle_a^{(M)}|N\rangle_b^{(M)},\\
|N\rangle_a^{(M)} &=&c_N^{(M)}(\pdm)^N|0\rangle,
\;\;|N\rangle_b^{(M)} =c_N^{(M)}(\qdm)^N|0\rangle,\nonumber\\
\end{eqnarray}
with normalization constant $c_N^{(M)}=(N M!/M^N)^{-\frac{1}{2}}$,
while the separable states are given by
\begin{equation}\label{Phiprod}
|\Phi_{N}\rangle = |\phi^{(2N)}\rangle_a|\phi^{(2N)}\rangle_b,
\end{equation}
where $|\phi^{(2N)}\rangle_{a,b}$ are separable states in the
modes $\ad_k$
and $b^{\dagger}_k$ respectively.\\
After the input state has evolved under the Hamiltonian
$H_{\varphi}$ into the state $|\Psi^{(M)}_N(\varphi)\rangle =
e^{iH_c\varphi} |\Psi^{(M)}_N\rangle$ an observable $O$ is used as
an estimator to determine the parameter $\varphi$ to a precision
given by \eqref{varest}. Instead of working in the Schr\"odinger
picture of state evolution it turns out to be more convenient to
tackle the problem in the Heisenberg picture, where $O$ evolves
according to $O \rightarrow
O^{\prime}=e^{-iH_c\varphi}Oe^{iH_c\varphi}$. We are interested in
the phase sensitivity for small $\varphi$, so that we can expand
$\eqref{varest}$ in powers of $\varphi$, arriving at
\begin{eqnarray}
O(\varphi) &=& O-i\varphi[H_c,O]-\frac{1}{2}\varphi^2(H_c^2O +
OH_c^2 - 2
H_cOH_c)\nonumber\\
&&+\mathcal{O}(\varphi^3).
\end{eqnarray}
If the input state $|\Psi^{(M)}_N\rangle$ is an eigenvector of $O$
with eigenvalue $0$, we obtain the following simple expressions for
$\langle O\rangle$ and $\mathrm{Var}(O)$

\begin{eqnarray}
\left|\frac{\partial}{\partial \varphi}\langle O\rangle \right|^2 &=& 4 \varphi^2 |\langle H_cOH_c \rangle |^2 +\mathcal{O}(\varphi^3),\\
\mathrm{Var}(O)&=&\varphi^2 \langle H_c O^2 H_c\rangle +
\mathcal{O}(\varphi^3),
\end{eqnarray}

so that the phase fluctuation $(\delta \varphi)^2$ simplifies to

\begin{equation}\label{deltaphisimple}
(\delta \varphi)^2 = \frac{\langle H O^2 H\rangle}{4 |\langle HOH
\rangle |^2} + \mathcal{O}(\varphi)
\end{equation}
An observable fulfilling this property is $O = (n^{(-)}_M)^2$,
where $n^{(-)}_M = \frac{1}{2}(\pdm \ppm -\qdm\qm).$\\
The commutation relations for $\pdm$ and $\qdm$ \eqref{commrel}
imply that in the limit of infinitely many modes $M \rightarrow
\infty$ the operators $\pdm$ and $\qdm$ become bosonic. We will thus
start out with a scenario where the input states are in the bosonic
limit and then turn our attention to a setting which is far from the
bosonic limit.

\subsubsection{Bosonic limit}\label{quasi-bosonic}
In this section we will consider the scenario $M \rightarrow
\infty$, i.e. we are in the bosonic limit, where the limit is
taken for the expectation values of the operators. We will
consider a coupling of the form $H_c=\varphi H_{\infty}$ where

\begin{equation}\label{Hinf}
H_{\infty}=\frac{1}{2}(a^{\dagger}_{\infty}b_{\infty}+a_{\infty}b^{\dagger}_{\infty}),
\end{equation}
and measure $(n^{(-)}_{\infty})^2$.\\
We start deriving the best precision for unpaired states  using
\eqref{estviavar}. We will use a finite $M$ for input state,
coupling Hamiltonian and measurement and then take the limit $M
\rightarrow \infty$. To be precise, the calculation will be done
for $H_M = \frac{1}{2}(a^{\dagger}_{M}b_{M}+a_{M}b^{\dagger}_{M})$
and $(n^{(-)}_{M})^2$. Then $\lim_{M \rightarrow \infty}\langle
\phi_{a,b}|H_{M}|\phi_{a,b}\rangle =0$ due to the conservation of
particle number. Hence, $\lim_{M \rightarrow \infty}
\mathrm{Var}(H_{M})=\lim_{M \rightarrow \infty}\langle
H_{M}^2\rangle = \lim_{M \rightarrow \infty}(\langle \pdm
\ppm\rangle + \langle \qdm \qm\rangle)^2 = 0,$ as $\langle \pdm
\ppm \rangle = \frac{1}{M}\sum_{kl}|\langle \Pd_k P_l \rangle|^2
\leq N/M$, where the last inequality results from the bound of the
pairing measure on unpaired states Thm.
\ref{measure-bound-separable-state}. The same holds for $\langle
\qdm \qm \rangle$. Hence, in the limit $M \rightarrow \infty$ the
variance of $H_{\infty}$ vanishes. For the setting of paired
states, however, we can obtain the following result:

\begin{theorem}
For paired input states, the interferometer depicted in Fig
\ref{complex-setup} allows to estimate the coupling parameter
$\varphi$ to a precision
\begin{equation}
(\delta \varphi)_{\inf}^2 = \frac{1}{2N^2}.
\end{equation}
\end{theorem}

\begin{proof}
Consider an interferometric setup depicted in Fig.
\ref{complex-setup}, where the $2N$-particle input state and the
coupling Hamiltonian are defined in eqs. \eqref{paired-in} and
\eqref{Hinf} respectively. We will again use a finite $M$ for
input state, coupling Hamiltonian and measurement and then take
the limit $M \rightarrow \infty$, i.e. we use $|\Psi_{in}\rangle
=|N\rangle_a^{(M)}|N\rangle_b^{(M)},$ $H_M =
\frac{1}{2}(a^{\dagger}_{M}b_{M}+a_{M}b^{\dagger}_{M})$ and
$(n^{(-)}_{M})^2$. Making use of the relations
\begin{eqnarray}
\ppm |N\rangle_a^{(M)} &=& \alpha_N |N-1\rangle_a^{(M)}\\
\pdm |N\rangle_a^{(M)} &=& \alpha_{N+1} |N+1\rangle_a^{(M)}
\end{eqnarray}
where $\alpha_N = \sqrt{N(1-(N-1)/M)}$, a lengthy but
straightforward calculation leads to $ (\delta \varphi)_M^2 =
\frac{1}{2\alpha_{N+1}^2 \alpha_N^2} + \mathcal{O}(\varphi)$ using
\eqref{deltaphisimple}. Taking the limit $M \rightarrow \infty$
leads to the result of the theorem.
\end{proof}

\subsubsection{Interferometry far from the bosonic
limit}\label{farfromboson} In the preceding section we have
studied the power of paired states in the bosonic limit. As the
power of bosonic particles for interferometry has been known for
quite a while, the use of paired states where the fermionic nature
of the particles survives might be a more interesting question. In
this section we will show that even far from the bosonic limit
paired states can achieve a precision gain of
order $N$ for quantum metrology.\\
We will study a coupling Hamiltonian of the form  $H_c=\varphi
H_{F}$ where
\begin{equation}\label{HF}
H_{F}= \sum_{k=1}^{\infty} \Pd_k Q_k + P_k \Qd_k.
\end{equation}
First, we will give a bound for the phase sensitivity achievable by
using product states at the input:

\begin{theorem}
Using product states of $2 N$ particles as input states for the
interferometric setting depicted in Fig. \ref{complex-setup}, the
phase $\varphi$ of the coupling Hamiltonian $H_c=\varphi H_{F}$,
where $H_F$ is defined in \eqref{HF} can be measured to a
precision $(\delta \varphi)^2 \geq 1/(16 N)$.
\end{theorem}

\begin{proof}
For every product state of the form \eqref{Phiprod} $\langle
H_F\rangle =0$ due to particle number conservation. Hence,
$\mathrm{Var}(H_F)=\langle H_F^2 \rangle$. We will bound this
expectation value:
\begin{eqnarray*}
\langle H_{F}^2\rangle &=& \sum_{k\neq l}\langle \Pd_k P_l\bracket
\Qd_l
Q_k\rangle + c.c. + \sum_k \langle \Pd_k P_k\bracket Q_k \Qd_k\rangle\\
&\leq& 2 \left(\sum_{k\neq l}|\langle \Pd_k P_l
\rangle|^2\right)^{1/2}\left(\sum_{k\neq l}|\langle \Qd_k Q_l
\rangle|^2\right)^{1/2}\\
&&+ \sum_k \langle \Pd_k P_k \bracket Q_k \Qd_k \rangle + c.c.
\end{eqnarray*}
From Lemma \ref{M-on-product} we know that $\left(\sum_{k\neq
l}|\langle \Pd_k P_l \rangle|^2\right)^{1/2}\leq \sqrt{N}$. Further,
$\langle P_k \Pd_k \rangle = \langle 1-(n_k-n_{-k})^2-n_kn_{-k}
\rangle \leq 1$ and $\sum_k \langle \Pd_k P_k \rangle \leq N$. Thus
$\mathrm{Var}(H_{pq})\leq 2\sqrt{N}\sqrt{N} + 2N = 4N$ which leads
immediately to our result via \eqref{estviavar}.
\end{proof}

This bound can be beaten by a factor of $\sqrt{N}$ using paired
states. A lengthy but straightforward calculation leads to the
following result:

\begin{theorem}\label{interferometer-paired}
Using paired states of the form \eqref{paired-in} as input states
for the interferometric setting depicted in fig.
\ref{complex-setup}, the phase $\varphi$ of the coupling
Hamiltonian $H_c=\varphi H_{F}$, where $H_F$ is defined in
\eqref{HF} can be measured to a precision
\begin{equation}
(\delta \varphi)^2 = \frac{M(M-1)}{8N(M-N)(M - 1 + M N -N^2)}.
\end{equation}
\end{theorem}

This theorem implies $(\delta \varphi)^2 \sim 1/N^2$ for all $M
\geq 2 N$. In conclusion we have shown that paired states are a
resource for quantum metrology. Theorem
\ref{interferometer-paired} is the main result of this section. We
have remarked already at the beginning of this chapter that it has
been proven before that the Heisenberg limit can be achieved using
fermionic particles \cite{PhysRevLett.56.1515}. However, these
states were constructed in an abstract way, while we prove that
the BCS states that can be created easily in an experimental setup
are a very powerful resource for quantum metrology.

%% file: Conclusion-s1.tex
\section{Application to experiments and conclusion}\label{Conclusion-section}
In summary, we have developed a pairing theory for fermionic states. We have
given a precise definition of pairing based on a minimal list of natural
requirements. We have seen that pairing is not equivalent to entanglement of
the whole state nor of its two-particle reduced density operator but
represents a different kind of quantum correlation. Within the
framework of fermionic Gaussian states we could solve the pairing problem
completely. For number conserving states we have given sufficient conditions
for the detection of pairing that can be verified by current experimental
techniques, e.g. via spatial noise correlations
\cite{PhysRevA.70.013603,Rom,greiner:110401} and prescribed a systematic way
to construct complete families of pairing witnesses.

To shed some light on the pairing debate \cite{partridge:020404,partridge_27,
  zwierlein_response, patridge_response}, we would need access to the
proportionality factor linking the quantity plotted in Fig.~4 of
\cite{partridge:020404} to the local pair correlation correlation function
$G_2(r,r)=\langle
\Psi_{\downarrow}^{\dagger}(r)\Psi_{\uparrow}^{\dagger}(r)\Psi_{\uparrow}(r)\Psi_{\downarrow}(r)
\rangle$.

Another important point of our work is the utility of fermionic states for
quantum metrology. While it has been shown that, in principle, fermionic
states can achieve the Heisenberg limit for precision measurements in a
Ramsey-type interferometer \cite{PhysRevLett.56.1515}, we could prove the
usefulness of states that are available in the laboratory. Furthermore, the
optimal precision for the Ramsey-type setup is proportional to the pairing
measure introduced from an intuitive picture in Secs.~\ref{Gaussian-section}
and \ref{Number-conserving-section}. This endows the measure with an
operational meaning.\\
The results we have presented are just a first step in understanding pairing
and its relation to other types of quantum correlations.

We hope that the pairing theory we have developed will help to get a
better understanding of correlated many-body systems, and can provide a new
perspective on quantum correlationsa and may serve as a
starting point for further inquiries.

For example, one might attempt a finer characterization of pairing, e.g.,
$\sum_{k=1}^2P_k^\dag\vac$ and $\sum_{k=1}^MP_k^\dag\vac$ represent paired
states of rather different nature: it would be interesting to develop
witnesses or measures which allow to determine over how many modes the pairs
in a given states extend and to relate these differences to applications in
metrology ore elsewhere?. Moreover, the theory we developed has been concerned
with finitely many modes only and it is an obvious question whether
generalizing to an infinite dimensional single-particle space gives rise to
new phenomena.

Up to now we have concentrated on fermionic states. But the question of
pairing in bosonic systems might be equally interesting and relevant for
recent experiments \cite{RSK+08}.

What about higher-order correlations? The set of unpaired states contains both
separable and highly correlated states. This is, for example, reflected in the
fact that there are unpaired states which can be transformed to paired ones by
single-mode particle number measurements (e.g.,
$(\ad_1\ad_2\ad_3+\ad_4\ad_5\ad_6)\vac$ by measuring particle number in mode
$b = a_3+a_6$). A theory of higher-order correlated states could be developed
along the lines discussed here, e.g., by changing the set of observables on
which the states are compared to uncorrelated ones and defining as $n$th-order
correlated those states whose expectation values on $n$th-order observables
cannot be reproduced by $m<n$-correlated states.

Tools and methods from entanglement theory have been very useful in analyzing
pairing. One very important such tool, however, is missing: positive
maps, that is, transformations which do not correspond to a physical
operations but nevertheless, when applied to a subsystem in a separable state
with the rest, map density operators to (unnormalized) density operators and
thus provides strong necessary conditions for separability. Finding an analogy
might prove very useful for the analysis of many-body correlations. Another
importat object in the theory of entanglement is the set of LOCC operations
(local operations and classical communication), i.e., the operations that
cannot create entanglement. In the case of pairing, the analogous set would
contain passive operations and discarding modes. Are there other physical
transformations that cannot create pairing? Do paired states, then, possibly
allow to implement such transformations similar to entanglement enabling
non-LOCC operations?

\begin{acknowledgments}
  We acknowledge support by the Elite Network of Bavaria QCCC, the DFG
  within SFB 631, the DFG Cluster of Excellence Munich-Centre for Advanced
  Photonics, QUANTOP and the Danish Natural Science Research
  Council (FNU).
\end{acknowledgments}


%% file: Appendix-s1.tex
\section{Useful properties of separable states}\label{usefull}
We give two technical lemmas which involve useful properties of
product states.

\subsection{Bound of $\ad_i\ad_ja_k a_l +h.c.$ on separable
states}\label{bound-simple} In this section we prove a bound of a
some special two-body operator on product states:
\begin{lemma}\label{simplewitness}
Let $\rho \in S_{sep}$ be a separable state. Then
\begin{equation}
|\mathrm{tr}[(\ad_i\ad_ja_k a_l +h.c.)\rho_s]|\leq 1/2
\end{equation}
\end{lemma}

\begin{proof}
Let $\mathcal{H}_{ijkl}$ be the Hilbert space spanned by $\ad_i,
\ad_j, \ad_k, \ad_l$ and define $A_{ijkl}=\ad_i\ad_ja_k a_l
+h.c.$. Then
$\mathrm{tr}[A_{ijkl}\rho]=\mathrm{tr}[\rho^{(ijkl)}A_{ijkl}]$,
where $\rho^{(ijkl)}=\sum_{n=0}^4 \beta_{ijkl}^{(n)}|n\bracket n|$
is a mixed separable state according to Lemma \ref{le:sep-trace},
and $|n\rangle$ denotes the occupation number basis for the
subspace $\mathcal{H}_{ijkl}$. It is easily checked that
$A_{ijkl}$ can have non-vanishing expectation value only for the
two-particle state $|2\rangle = \left(\sum_{r = i,j,k,l}\mu_r
\ad_r \right)\left(\sum_{s=i,j,k,l}\nu_s \ad_s \right)|0\rangle.$
Using $|2 \mathrm{Re}(ab)| \leq |a|^2+|b|^2$ for any complex
numbers $a, b$ and the normalization conditions $\sum_r|\mu_r|^2=
\sum_r|\nu_r|^2=1$, one arrives at $ |\mathrm{tr}[A_{ijkl}\rho]|=2
|\mathrm{Re} [(\mu_i \nu_j - \mu_j \nu_i)(\mu_k \nu_l - \mu_l
\nu_k)^*]| = (|\mu_i|^2+ |\mu_j|^2)(|\mu_k|^2+ |\mu_l|^2)+
(|\nu_i|^2+ |\nu_j|^2)(|\nu_k|^2+ |\nu_l|^2) \leq 0.25 + 0.25 =
0.5$.
\end{proof}

\subsection{Expectation values of one-and two-body operators for separable states}
In this section we will prove that the one-and two-body operators
for separable states can be expressed in terms of matrix elements of
projectors.

\begin{lemma}\label{exp-as-proj}
Let $\rho \in \mathcal{S}_{sep}^{(N)}$ be a pure separable state.
Then
\begin{eqnarray}
\langle n_i\rangle &=& P_{ii}\\
 \langle \ad_i\ad_j a_k a_l\rangle_{\rho}
&=& (P\otimes P)_{(ij)(lk)}-(P\otimes
P)_{(ij)(kl)},\nonumber\\\label{aaaa}
\end{eqnarray}
where $P=P^2=P^{\dagger}$ is a projector of rank $N$.
\end{lemma}

\begin{proof}
Consider $M$ modes. We go into the basis where the pure separable
state is of the form $|\Phi\rangle = \prod_{i=1}^N
\ad_{\alpha_i}\vac.$ In this basis $
\langle\ad_{\alpha_i}\ad_{\alpha_j}a_{\alpha_k}a_{
\alpha_l}\rangle_{|\Phi\rangle} =
\delta_{il}\delta_{jk}-\delta_{ik}\delta_{jl}$, i.e. \ref{aaaa} for
$P = \mathrm{Id}_N$, where $\mathrm{Id}_N = \mathds{1}_N \bigoplus
0_{M-N} \in \mathds{C}^{M \times M}$. Now let $\ad_i = \sum_k
U_{ik}\ad_{\alpha_k}$. Then
\begin{eqnarray*}
\langle\ad_{i}\ad_{j}a_{k}a_{l}\rangle_{\rho_s^{(N)}} &=& (U\
\mathrm{Id}_N U^{\dagger})\otimes (U \mathrm{Id}_N
U^{\dagger})_{(ij)(lk)}-\nonumber\\
&&(U \mathrm{Id}_N U^{\dagger})\otimes
(U \mathrm{Id}_N U^{\dagger})_{(ij)(kl)}\nonumber\\
&=&(P\otimes P)_{(ij)(lk)}-(P\otimes P)_{(ij)(kl)},
\end{eqnarray*}
and $P$ is a projector of rank $N$.\\
For the one-particle operators, we obtain $\langle n_i\rangle =
P_{ii},$ as $(N-1)\langle n_i\rangle = \sum_{j\neq i} \langle n_i
n_j \rangle = \sum_{j \neq i} P_{ii}P_{jj} - |P_{ij}|^2 =\sum_j
P_{ii}P_{jj}-P_{ij}P_{ji}= \mathrm{tr}[P]P_{ii}-(P^2)_{ii} =
(N-1)P_{ii}$.
\end{proof}

\section{Proof of Lemmas \ref{Convex-Sep} and
\ref{Convex-all}}\label{convex-proof}
\subsection{Proof of Lemma \ref{Convex-Sep}}
\begin{proof}
As $H_{1\pm}^{(p)}, H_{2\pm}^{(p)}, H_{3\pm}^{(p)}$ are built up
of operators that are the product of at most two creation and
annihilation operators, we can prove the lemma for separable
states.  In the first step, we will show that the three operators
are positive on all separable states. Then we will show that all
states within the set bounded by $H_{1\pm}^{(p)}, H_{2\pm}^{(p)},
H_{3\pm}^{(p)}$ correspond to a separable state. Finally we will
show there exist states that are detected as paired by
$H_{1\pm}^{(p)}$ and $H_{3\pm}^{(p)}$. Positivity of
$H_{2\pm}^{(p)}$ on all number conserving states will be shown in
the proof of Lemma \ref{Convex-Sep} following below.\\
To show positivity of $H_{1\pm}^{(p)}, H_{2\pm}^{(p)},
H_{3\pm}^{(p)}$ it is sufficient to show the positivity for pure
separable states, as the result for mixed states follows from
convexity. From now on, let $\rho \in
\mathcal{S}_{sep}^{(N)}.$\\
{$\mbox{\textbf{tr}}\mathbf{[H_{1\pm}^{(p)} \rho] \geq 0}$}\\
In Lemma \ref{exp-as-proj} (see appendix \ref{usefull}) we have
shown that the expectation values of number conserving one-and
two-body operators can be expressed in terms of matrix elements of
projectors. Let $P$ be the rank $N$ projector such that $\langle
\ad_i\ad_j a_k a_l\rangle_{\rho} = (P\otimes
P)_{(ij)(lk)}-(P\otimes P)_{(ij)(kl)}$, $\langle n_i \rangle =
P_{ii}$ and let $\tilde{P}=P|_{k -k, l, -l}$ the $4 \times 4$
principal submatrix of $P$ where the indices run over $k, -k, l,
-l.$ Then we have the following inequalities:
\begin{eqnarray}
\langle n_k n_{-k}\rangle &=& P_{kk}P_{-k-k}-|P_{k-k}|^2\nonumber\\
&\leq&
\frac{1}{2}\left(|P_{kk}|^2+|P_{-k-k}|^2\right)-|P_{k-k}|^2,
\end{eqnarray}
\begin{eqnarray}
\lefteqn{|\langle \ad_k \ad_{-k}a_{-l}a_l +h.c.\rangle| }\nonumber\\
&&= 2 |\mbox{Re}(P_{kl}P_{-k-l}-P_{k-l}P_{-kl})|\\
&&\leq 2 (|P_{kl}||P_{-k-l}|+|P_{k-l}||P_{-kl}|)\nonumber\\
&&\leq(|P_{kl}|^2+|P_{-k-l}|^2+|P_{k-l}|^2+|P_{-kl}|^2).
\end{eqnarray}
These results imply
\begin{equation}
\mathrm{tr}[\rho H_{1\pm}^{(p)}]\geq
\frac{1}{2}\mathrm{tr}[\tilde{P}-\tilde{P}^2]+|P_{k-k}|^2+|P_{l-l}|^2
\end{equation}
We use the inclusion principle \cite{H&J}, stating that the
eigenvalues of a $r \times r$ principal submatrix $M_r$ of a $n
\times n$ Hermitian matrix $M$ fulfill $\lambda_k(M)\leq
\lambda_k(M_r)\leq \lambda_{k+n-r}(M),$ where the eigenvalues are
arranged in increasing order. As $P$ is a projector, we have $0
\leq \lambda_k(P)\leq \lambda_k(\tilde{P})\leq
\lambda_{k+M-r}(P)\leq 1$. Hence,
\begin{equation}
\mathrm{tr}[H_{1\pm}^{(p)}\rho]\geq
\frac{1}{2}\mathrm{tr}[\tilde{P}-\tilde{P}^2]\geq
\frac{1}{2}\sum_k \lambda_k(\tilde{P})(1-\lambda_k(\tilde{P}))\geq
0
\end{equation}
$\mbox{\textbf{tr}}\mathbf{[H_{2\pm}^{(p)} \rho] \geq 0}$\\
Define $O_1 = n_kn_{-k}+n_ln_{-l}\geq 0$, $O_2^{\pm} = 1\pm\ad_k
\ad_{-k}a_{-l}a_l + h.c.\geq 0$. Then
$H_{2\pm}^{(p)}=O_1O_2^{\pm},$ and as $[O_1, O_2^{\pm}]=0$ we
conclude that $H_{2\pm}^{(p)} = O_1O_2^{\pm} \geq
0$.\\
$\mbox{\textbf{tr}}\mathbf{[H_{3\pm}^{(p)} \rho] \geq 0}$\\
We will need the lemma \ref{le:sep-trace}:
$\mathrm{tr}[H_{3\pm}^{(p)} \rho]= \mathrm{tr}[H_{3\pm}^{(p)}
\rho_{kl}],$ where $\rho_{kl}= \sum_{n=0}^4 \beta_n|n\bracket n|$,
$\beta_n \geq 0,$ $\sum_{n=0}^4 \beta_n=1$ and $|n\rangle, n=0,
\ldots, 4$ are separable $n$-particle states. Let $\langle
H_{3\pm}^{(p)}\rangle_n = \langle n|H_{3\pm}^{(p)}|n \rangle.$
Then a straightforward calculation leads to $ \langle
H_{3\pm}^{(p)} \rangle_0 = 1,$ $\langle H_{3\pm}^{(p)} \rangle_1 =
\frac{1}{2}$, $\langle H_{3\pm}^{(p)} \rangle_2 = \frac{1}{2}$,
$\langle H_{3\pm}^{(p)} \rangle_3 = 0$ and $\langle H_{3\pm}^{(p)}
\rangle_4 =  0$. Linearity of the trace implies $\mathrm{tr}[H_{3\pm}^{(p)} \rho]\geq 0.$\\
Hence, all separable states lie within the set bounded by the
planes defined by the witness operators $H_{1\pm}^{(p)},
H_{2\pm}^{(p)}, H_{3\pm}^{(p)}.$ \\
Next, we show that each point within the polytope
$C^\mathrm{unpaired}$ corresponds to a separable state. As
$S_{sep}$ is convex, it is sufficient to check that for every
extreme point of $C^\mathrm{unpaired}$ there exists a separable
state. This is indeed the case: The extreme points of
$C^\mathrm{unpaired}$ are $(0,0,0)$, $(2,0,0)$, $(4,2,0)$,
$(2,1/2,\pm1/2)$ which correspond for example to the separable
states $\vac$, $\ad_k \ad_l \vac$, $\ad_k \ad_{-k}\ad_l
\ad_{-l}\vac$ and $(\ad_k + \ad_l)(\ad_{-k}\pm\ad_{-l})/2\vac$
respectively.\\
It remains to show that $H_{1\pm}^{(p)}$ and $H_{3\pm}^{(p)}$ are
pairing witnesses. Define $|\Psi\rangle =
\frac{1}{\sqrt{2}}(\ad_k\ad_{-k}+\ad_l\ad_{-l})\vac.$ Then
$\mathrm{tr}[H_{1\pm}^{(p)} |\Psi\bracket
\Psi|]=\mathrm{tr}[H_{3\pm}^{(p)} |\Psi\bracket \Psi|]=-1$.
\end{proof}

\subsection{Proof of Lemma \ref{Convex-all}}
\begin{proof}
It is sufficient to prove the lemma for $\rho \in
\mathcal{S}(\mathcal{A}_N)$, as the result for a general number
conserving state follows from convexity.\\
$\mbox{\textbf{tr}}\mathbf{[H_1 \rho] \geq 0}$\\
$H_1 = \frac{1}{2}(n_k-n_{-k})^2+\frac{1}{2}(n_l-n_{-l})\geq 0$\\
$\mbox{\textbf{tr}}\mathbf{[H_{2\pm} \rho] \geq 0}$\\
Shown in the proof of Thm. \ref{Convex-all}.\\
$\mbox{\textbf{tr}}\mathbf{[H_{3\pm} \rho] \geq 0}$\\
Let $\rho_{kl}=\sum_{n=1}^4 \beta_n |0\bracket n|$ be the reduced
density operator in the modes $\pm k, \pm l.$ We can rewrite
$H_{3\pm}$ in the form $H_3 = 2 -
\frac{1}{2}(n_k+n_{-k}+n_l+n_{-l})(1\mp (\ad_k\ad_{-k}a_{-l}a_l +
h.c.)).$ Defining $O_{2\mp} =1\mp (\ad_k\ad_{-k}a_{-l}a_l +
h.c.),$ we obtain $\langle O_{2\mp} \rangle_0=\langle
O_{2\mp}\rangle_1=\langle O_{2\mp} \rangle_3=\langle O_{2\mp}
\rangle_4=1, \langle O_{2\mp}\rangle_2 \leq 2$. This implies $
\mathrm{tr}[\rho H_{3\pm}] \geq 4-(\beta_1 + 2\beta_2 + 3\beta_3 +
4\beta_4) \geq 4-\sum_{n=0}^4 n \beta_n =4-
\mathrm{tr}[(n_k+n_{-k}+n_l+n_{-l})\rho_{kl}]\geq 0$\\
As in the proof of Lemma \ref{Convex-Sep} it remains to show that
the extreme points of $C^\mathrm{all}$ correspond to some
fermionic state. It has been shown in the proof of Lemma
\ref{Convex-Sep} that $(0,0,0)$, $(2,0,0)$ and $(4,2,0)$ can be
reached by some separable state. The remaining two extreme points,
$(2, 1, \pm 1)$ correspond for example to the state
$\frac{1}{\sqrt{2}}(\ad_k\ad_{-k}+\ad_l\ad_{-l})\vac$.
\end{proof}

\section{Proof of Thm. \ref{witness-Nc}}\label{Proof-witness}
In this section we will provide all the details leading to Thm.
\ref{witness-Nc}, starting with the bound on separable states:

\begin{lemma}\label{bound-on-ps}
Let $\rho \in S_{sep}$ with $\tr[N_{op}\rho]=N$ and let $H(\{v_k\})$
be as in Thm. \ref{witness-Nc}. Then
$\mathrm{tr}[H(\{v_k\})\rho_s]\geq 0$.
\end{lemma}

\begin{proof}
The operator $H_k$ acts non-trivially only on the modes $\ad_k,
\ad_{-k},\ad_{k+M}, \ad_{-(k+M)}$. Denote by $\rho_k$ the reduced
density operator obtained when tracing out all but these four
modes. According to Lemma \ref{le:sep-trace}, $\rho_k =
\sum_{n=0}^4 \beta_k^{(n)}|n\bracket n|$ is a convex combination
of separable $n$-particle states $|n\bracket n|$. We proved in
Lemma \ref{simplewitness} that $|\mathrm{tr}[(\Pd_k P_{k+M}+h.c.)
\rho_s ]|\leq 1/2$. Hence $\langle H(\{v_k\})\rangle \geq
2N(1-\epsilon)- 2 \sum_{k=1}^{M}(|v_k|^2 \langle N_k\rangle +
\beta_k^{(2)})$. Now $2 \beta_k \leq \langle N_k \rangle \leq 4$
and $|v_k|^2 \leq 1-\epsilon$ so that $|v_k|^2 \langle N_k\rangle
+ \beta_k^{(2)}) \leq 4$. Due to the particle number constraint
$\sum_{k=1}^M \langle N_k \rangle = N$ this value can be taken for
$k=1,\ldots, N/4$. Hence, $-2 \sum_{k=1}^{M}(|v_k|^2\langle
N_k\rangle + \beta_k^{(2)})\geq -2\cdot 4(1-\epsilon)N/4$ so that
$\langle H(\{v_k\}) \rangle \geq 0$.
\end{proof}

To show the witness character of $H(\{v_k\})$ we also have to
prove that there exists a BCS state that is detected by the
Hamiltonian. We will need the following theorem about the
distribution described by the $|\lambda_N|^2$ in \eqref{ncBCS}:

\begin{theorem}\label{normal-dist-proof}
Let $|\Psi_{Gauss}\rangle = \sum_{N=0}^M
\lambda_N^{(M)}|\Psi_{BCS}^{(N)}\rangle$ like in \eqref{BCS-var},
\eqref{BCS-N}. If $\sum_{k=1}^M |u_k|^2|v_k|^2 =
\mathcal{O}(N^{\gamma})$ for some $\gamma >0$, then in the limit
$N \rightarrow \infty$ the $|\lambda_N|^2$ converge to a normal
distribution,
\begin{equation}
|\lambda_N|^2 = \frac{1}{\sqrt{2
\pi}\sigma_{\bar{N}}}\exp\left[-\frac{(N-\bar{N})^2}{2
\sigma_{\bar{N}}^2} \right],
\end{equation}
where $2 \bar{N}= 2\sum_{k=1}^M |v_k|^2$ is the mean particle number
for the variational state, and the variance is given by
$\sigma_{\bar{N}}^2= 4 \sum_{k=1}^M |v_k|^2 |u_k|^2$.
\end{theorem}

\begin{proof}
For the proof we will need a theorem from probability theory known
as Lyapunov's central limit theorem \cite{Billingsley95}:

\begin{theorem}\label{Lyapunov}(Lyapunov's central limit theorem)\\
Let $ X_1, X_2,\dots$ be independent random variables with
distribution functions $ F_1,F_2,\dots$, respectively, such that $
EX_n=\mu_n$ and $ \operatorname{Var}X_n=\sigma_n^2<\infty$, with
at least one $ \sigma_n>0$. Let $S_n = X_1+\cdots+X_n$ and
$s_n=\sqrt{\operatorname{Var}(S_n)} =
\sqrt{\sigma_1^2+\cdots+\sigma_n^2}.$ If the Lyapunov condition
$$\frac{1}{s_n^{2+\delta}}\sum_{k=1}^n E\vert
X_k-\mu_k\vert^{2+\delta} \xrightarrow[n\rightarrow\infty]{} 0$$
is satisfied for some $ \delta>0$ then the normalized partial sums $
\frac{S_n - ES_n}{s_n}$ converge in distribution to a random
variable with normal distribution $ N(0,1)$
\end{theorem}

Consider the observables $X_k= n_{k} + n_{-k}, k=1, \ldots, M$,
where $n_{\pm k}=0,1$ is the number of particles with quantum
numbers $\pm k$ respectively. The $X_k$ can be considered as
classical random variables since they commute mutually. In the
variational BCS-state the random variable $S_M = \sum_{k=1}^M X_k$
is distributed according to the probability distribution

\begin{eqnarray}\label{lambda-C}
\lefteqn{P(S_M=2N) =}\nonumber\\
&&\sum_{k_1 < \ldots < k_M = 1}^M|v_{k_1}|^2\ldots |v_{k_N}|^2
|u_{k_{N+1}}|^2\ldots |u_{k_M}|^2\nonumber\\
&=&\left(\prod_k |u_k|^2 \right)\sum_{k_1 < \ldots < k_M =1
}^M\frac{|v_{k_1}|^2\ldots |v_{k_N}|^2}{|u_{k_1}|^2\ldots
|u_{k_N}|^2}\\
&=& \frac{|C|^2}{(N!)^2 |C_N|^2} = |\lambda_N^{(M)}|^2.
\end{eqnarray}

With the help of Thm. \ref{Lyapunov} applied to the random
variable $S_M$ we can now complete the proof of Thm.
\ref{normal-dist-proof}, i.e. show that $\lambda_N^{(M)}$
converges to a normal distribution for large $M$. We start
calculating the expectation value $\mu_k$ of $X_k$. For a
BCS-state, $X_k = 0,2$, as particles with quantum numbers $\pm k$
always appear in pairs. As $P(X_k=0)=|u_k|^2, P(X_k=2)=|v_k|^2$,
we get $\mu_k=2 |v_k|^2$ and $E(S_M)=2 \sum_k |v_k|^2 \equiv 2
\bar{N}$. For calculating the variance, note that $X_k^2= n_{k}^2+
n_{-k}^2 + 2 n_{k}n_{-k}=0,4$, and $P(X_k^2=0)= |u_k|^2,\;
P(X_k^2=4)=|v_k|^2$. Hence

\begin{equation}
\mathrm{Var}(X_k)=4|u_k|^2|v_k|^2, \hspace{1cm} s_M^2=4\sum_k
|u_k|^2|v_k|^2.
\end{equation}

To apply the central limit theorem, we consider
$E(|X_k-\mu_k|^4)$. Using $P(|X_k-\mu_k|^4=\mu_k^4)=|u_k|^2$,
$P(|X_k-\mu_k|^4=(2-\mu_k)^4)=|v_k|^2$, and $\mu_k=2|v_k|^2$ we
arrive at

\begin{eqnarray}
E(|X_k-\mu_k|^4)&=&
16|u_k|^2|v_k|^2(|u_k|^6+|v_k|^6)\nonumber\\&\leq&
16|u_k|^2|v_k|^2(|u_k|^2+|v_k|^2)\nonumber\\
&=&16|u_k|^2|v_k|^2.
\end{eqnarray}

Setting $\delta=2$ in the Lyapunov condition, we obtain

\begin{equation}
\frac{1}{s_M^4}\sum_{k^=1}^M E(|X_k-\mu_k|^4) \leq\frac{4}{s_M^2} =
\mathcal{O}(N^{-\gamma}) \rightarrow 0,
\end{equation}

where we have applied the assumption of the theorem $\sum_{k=1}^M
|v_k|^2 |u_k|^2 = \mathcal{O}(N^{\gamma})$ in the last step. The
central limit theorem implies that $S_M$ converges to a normal
distribution with expectation values $2 \bar{N}= 2\sum_k |v_k|^2$
and variance $\sigma_{\bar{N}}^2 = 4\sum_k |v_k|^2|u_k|^2$.
\end{proof}

With this result at hand we can prove the following:

\begin{lemma}\label{bound-on-BCS}
Let $H(\{v_k\})$ and $|\Psi_{BCS}^{(N)}\rangle$ be defined as in
Thm. \ref{witness-Nc}. If $\epsilon > 18/\sqrt{\pi N}$, then
$$\langle \Psi_{BCS,sym}^{(N)}|H(\{v_k\})|\Psi_{BCS,sym}^{(N)}\rangle\ < 0.$$
\end{lemma}

\begin{proof}
We will use the correspondence of variational and number
conserving BCS states, deriving first a bound for $|\langle
H(\{v_k\})\rangle_{var} - \langle H(\{v_k\}) \rangle_{N}|$, where
$|\Psi_{Gauss}\rangle= \sum_{k=1}^{2M} \lambda_n
|\Psi_{BCS,sym}^{(n)}\rangle$ with $|\Psi_{BCS,sym}^{(n)}\rangle$
like in Thm. \ref{witness-Nc}. To do so, we will need that the
$|\lambda_n|^2$ are normally distributed. From $|u_k|^2 =
1-|v_k|^2 \geq \epsilon$ where $\epsilon > 18/(\sqrt{2 \pi \bn})$
and $\sum_{k=1}^{2M} |v_k|^2 = \bn$, it follows that $\sn =
\sum_{k=1}^{2M}|v_k|^2|u_k|^2 \geq \epsilon \bn =
\mathcal{O}(\sqrt{N})$. Hence, we know from Thm.
\ref{normal-dist-proof} that the $|\lambda_n|^2$ describe a normal
distribution around $\bn \approx N$ with standard deviation $\sn$.\\
Now, write $H(\{v_k\})= H_0 - 2W +2(1-\epsilon)\sum_{k=1}^M N_k$,
where
\begin{eqnarray*}
H_0 &=& -2\sum_{k=1}^M |v_k|^2N_k,\\
W &=& 2 \sum_k \Pd_k P_{k+M} + \Pd_{k+M}P_k.
\end{eqnarray*}
We start with a bound for $|\langle H_0\rangle_{var} - \langle
H_0\rangle_{N}|  \leq T_1 +T_2$, where
\begin{eqnarray}\label{Difference-two-sums}
T_1&=&\left|\sum_{\Delta \in [- \sigma^2_{\bar{N}},
\sigma^2_{\bar{N}}]}|\lambda_{N+\Delta}|^2(\langle
H_0\rangle_{N+\Delta}-\langle
H_0\rangle_{N})\right|,\nonumber\\
T_2&=&\left|\sum_{\Delta \notin [- \sigma^2_{\bar{N}},
\sigma^2_{\bar{N}}]}(|\lambda_{N+\Delta}|^2(\langle
H_0\rangle_{N+\Delta}-\langle
H_0\rangle_{N})\right|.\nonumber\\
\end{eqnarray}
A bound for $T_2$ can be easily derived noting that
\begin{equation}\label{DiffH0}
|\langle H_0 \rangle_{n}-\langle H_0 \rangle_{n'}| =8 \left|\sum_k
|v_k|^2 (\langle n_k \rangle_n - \langle n_k \rangle_{n'})\right|,
\end{equation}
and for $n=N+\Delta >= N$ we have $\sum_k |v_k|^2(\langle n_k
\rangle_n - \langle n_k \rangle_{N}) \leq \sum_k |v_k|^2\langle
n_k \rangle_n \leq n,$ as $|v_k|^2 \leq 1$. Hence,
\begin{eqnarray}
T_2 &\leq& 16\left|\sum_{\Delta \notin [0,
\sigma^2_{\bar{N}}]}|\lambda_{N+\Delta}|^2 |N + \Delta| \right|
\nonumber\\
&& \leq 8 \frac{\sigma_{\bar{N}}}{\sqrt{2 \pi}}e^{-\sigma_{N}^2/2}
+ 4 N (1-\mbox{Erf}(\sn/\sqrt{2}))\nonumber\\
&& \leq 8 \frac{\sigma_{\bar{N}}}{\sqrt{2 \pi}} + 4 N
(1-\mbox{Erf}(\sn/\sqrt{2})),
\end{eqnarray}
where we have approximated the sum by an integral in the second
step. For bounding $T_1$, we will show first that for $n = N+
\Delta$ where $\Delta \in [-\sn^2, \sn^2]$ we have $\langle n_k
\rangle_n - \langle n_k \rangle_{n-1}\geq 0$. Expanding the BCS
wave function
\begin{equation}
|\Psi_{BCS,sym}^{(n)}\rangle = C_n n! \sum_{j_1 < \ldots <
j_n=1}^{2M}\alpha_{j_1}\ldots \alpha_{j_n}\Pd_{j_1}\ldots
\Pd_{j_n}\vac,
\end{equation}
the expectation value of the number operator is easily calculated to
be $\langle n_k \rangle_n = |C_n|^2 (n!)^2 |\alpha_k|^2
S_k^{(n-1)}$, where
\begin{equation}
S_k^{(n)}= \sum_{j_1< \ldots <j_{n}=1\atop j_i \neq
k}^{2M}|\alpha_{j_1}|^2\ldots|\alpha_{j_{n}}|^2.
\end{equation}
If $0 < |v_k|^2 \leq 1- \epsilon$, there exists a lower bound on the
coefficients $|\alpha_k|^2 = |v_k|^2/\sqrt{1-|v_k|^2} \geq b
\;\forall\; k$. Then $S_k^{(n-1)}$ and $S_k^{(n-2)}$ are related via
$$S_k^{(n-1)}\geq b \frac{2M-(n-1)}{n-1}S_k^{(n-2)}.$$
In the proof of Thm. \ref{normal-dist-proof} we show that
\begin{equation}
\frac{|\lambda_n|^2}{|\lambda_{n-1}|^2}=\frac{|C_{n-1}|^2
((n-1)!)^2}{|C_n|^2 (n!)^2},
\end{equation}
resulting in
\begin{eqnarray}
\lefteqn{\langle n_k \rangle_n - \langle n_k \rangle_{n-1}
\geq}&&\nonumber\\&& \left(
b\frac{2M-(n-1)}{n-1}-\frac{|\lambda_n|^2}{|\lambda_{n-1}|^2}\right)|C_n|^2(n!)^2|\alpha_k|^2S_k^{(n-2)}.\nonumber\\
\end{eqnarray}
For $n= N + \Delta$ and $\Delta \in [- \sigma^2_{N},
\sigma^2_{N}]$ the normal distribution of the $|\lambda_n|^2$
implies $|\lambda_n|^2/|\lambda_{n-1}|^2 = \exp[(2\Delta-1)/(2
\sigma_{N}^2)]\leq e$. Hence,
$b\frac{2M-(n-1)}{n-1}-\frac{|\lambda_n|^2}{|\lambda_{n-1}|^2}
\geq 0 \Leftrightarrow b \geq
e\frac{n}{2M-(n-1)}>3\frac{n-1}{2M-(n-1)}$. For $M = q (n-1)$,
this is equivalent to $|\alpha_k|^2 \geq 3/(2q-1)$, which can be
achieved for $q\gg 1$. The last condition is satisfied, as we are
considering dilute systems, where $M \gg \bn$. Thus, $\langle n_k
\rangle_n - \langle n_k\rangle_{n-1}\geq 0$, implying $\sum_k
|v_k|^2 (\langle n_k \rangle_n - \langle n_k\rangle_{n-1})\leq 1$,
as $|v_k|^2\leq 1$. Using \eqref{DiffH0} and a telescope sum, we
conclude that
\begin{eqnarray}
T_1&\leq& 8 \left|\sum_{\Delta \in [- \sigma_{\bar{N}},
\sigma_{\bar{N}}]}(|\lambda_{N+\Delta}|^2 |\Delta|
\right|\nonumber\\
&\leq& 8 \left|\sum_{\Delta}(|\lambda_{N+\Delta}|^2 |\Delta|\right|
= 16 \frac{\sigma_{\bar{N}}}{\sqrt{2\pi}}.
\end{eqnarray}

Next, we derive the bound for the operator $W$. Its expectation
value is given by
\begin{equation}
\langle W \rangle_n = |C_n|^2 (n!)^2 2 \sum_k |\alpha_k|^2
\sum_{j_1 < \ldots < j_{n-1}\atop j_i \neq k,
k+M}^{2M}|\alpha_{j_1}|^2\ldots |\alpha_{j_{n-1}}|^2.
\end{equation}
For $n \in [N - \Delta, N + \Delta]$, we use the same argumentation
we have used for bounding $\langle n_k \rangle_n -\langle n_k
\rangle_{n-1}$, to obtain
\begin{equation}
\langle W \rangle_n - \langle W \rangle_{n-1} \leq 2.
\end{equation}
Further, $\langle n_k \rangle_n = \langle \Pd_k \Pd_{k-M} + h.c.
\rangle_n/2 + \langle n_k n_{k+M}\rangle_n$ due to the symmetry
$\alpha_k = \alpha_{k+M}$. Hence, $\langle W \rangle_n \leq 2 n$.
Thus, up to a factor of 2 we obtain the same bound as for $H_0$.
Putting all the pieces together we find that
\begin{eqnarray}
\lefteqn{|\langle H(\{v_k\})\rangle_{var}-\langle
H(\{v_k\})\rangle_{N}|
\leq}&&\nonumber\\
&&2(1-\epsilon)+\frac{72}{\sqrt{2\pi}}\sigma_{N}+12N\left(1-\mathrm{Erf}(\sigma_N/\sqrt{2})\right).\nonumber\\
\end{eqnarray}
In the limit of large $x$, the error function
$\mbox{Erf}(x/\sqrt{2\pi})$ can be approximated by the following
formula:
\begin{equation}
1-\mbox{Erf}(x/\sqrt{2})=
2\frac{\exp[-x^2/2]}{\sqrt{2\pi}}(x^{-1}-x^{-3}+\ldots).
\end{equation}
As $\sigma_N = \mathcal{O}(\sqrt{N})$, we conclude
$$12N
(1-\mbox{Erf}(\sigma_{\bn}/\sqrt{2}))\leq 24
\sigma_{\bn}\exp[-\sigma_{\bn}^2/2]/(\epsilon
\sqrt{2\pi})\rightarrow 0$$ for $N \gg 1$. A straightforward
calculation results in $\langle H(\{v_k\})\rangle_{var} =
-4N\epsilon$, leading immediately to the statement of the theorem.
\end{proof}

\section{Lemma for the proof of Thm.
\ref{measure-bound-separable-state}}\label{App-bound-on-prod}
\begin{lemma}\label{M-on-product}
Every pure separable state $\rho \in \mathcal{S}(\mathcal{A}_N)$
fulfills $\sum_{kl=1}^M|\langle \Pd_kP_l\rangle_{\rho}|\leq N/2$,
and this bound is tight.
\end{lemma}

\begin{proof}
Using Lemma \ref{exp-as-proj}, we obtain
\begin{equation}
\sum_{k,l=1}^{M}|\langle \Pd_k P_l \rangle|
=\sum_{k,l=1}^{M}|P_{kl}P_{-k-l}-P_{k-l}P_{-kl}|,
\end{equation}
where $P=P^2=P^{\dagger}$ and $\mathrm{tr}[P]=N$. Using the
triangle-inequality we get
\begin{eqnarray*}
\lefteqn{\sum_{k,l=1}^{M}|\langle \Pd_k P_l \rangle|
}&& \\
&\leq&\frac{1}{2} \sum_{k,l} \left(|P_{kl}|^2 +
|P_{-k-l}|^2+|P_{k-l}|^2+|P_{-kl}|^2 \right)\\ &=& \frac{1}{2}
\mathrm{tr}[P^2]=N/2.
\end{eqnarray*}
In the last step we have used the property that the sum of the
squares of a normal matrix is equal to the sum of squares of its
eigenvalues. Taking the square root we obtain
the bound of our claim.\\
The bound is tight, as $P = \mathds{1}_{2_N}$ implies
$\sum_{kl}|\langle \Pd_k P_l \rangle| = N/2$ which is obtained for
$|\Phi\rangle = \prod_{i=1}^N\ad_i\vac$.
\end{proof}

\section{Proof of Lemma \ref{Pairend-not-entangled}}\label{Ent-app}
\begin{proof}
Let $|i, j\rangle = \ad_i\ad_j |0\rangle$ and consider the subspace
spanned by the states $\{|k,-k\rangle, |l,-l\rangle, |k,l\rangle,
|k,-l\rangle, |-k,l\rangle, |-k,-l\rangle \}.$ In this basis, the
two-particle RDO $\rho^{(N)}_2$ of $|\Psi_{BCS}^{(N)}\rangle$  is of
the form

\begin{equation}
\rho^{(N)}_2 =\frac{1}{4+2a_1}\left(
                \begin{array}{ccc}
                  a_1 & a_2 & 0 \\
                  a_2 & a_1 & 0 \\
                  0 & 0 & \mathds{1}_{4} \\
                \end{array}
              \right),
\end{equation}

where $a_1 = (M-1)/(N-1), a_2=(M-N)/(N-1)$. The witness operator
$H_1^{(p)}$ of Thm. \ref{Convex-Sep} has a negative expectation
value on $\rho^{(N)}_2$, hence the state is paired in these modes.\\
For solving the entanglement question we will use the following
theorem \cite{ESBL-Indistinguishable} applicable to mixed fermionic
states of two particles each living on a single-particle Hilbert
space of dimension four:

\begin{theorem}\label{ent-of-particles-4}
Let the mixed state acting on $\mathcal{A}_4$ have a spectral decomposition $\rho = \sum_{i=1}^r
|\Psi_i \bracket \Psi_i|,$ where $r$ is the rank of $\rho$, and the
eigenvectors $|\Psi_i \rangle$ belonging to nonzero eigenvalues
$\lambda_i$ are normalized as $\langle \Psi_i|\Psi_j \rangle =
\lambda_i \delta_{ij}.$ Let $|\Psi_i \rangle = \sum_{a,b} wî_{ab}
\ad_a \ad_b |0\rangle$ in some basis, and define the complex
symmetric $r \times r$ matrix $C$ by
\begin{equation}
C_{ij} = \sum_{abcd} \epsilon^{abcd} w^i_{ab}w^j_{cd},
\end{equation}
which can be represented using a unitary matrix as $C=UC_dU^T,$ with
$C_d = \mbox{diag}[c_1, \ldots, c_r]$ diagonal and $|c_1| \geq
|c_2|\geq \ldots \geq |c_r|.$ The state has Slater number 1 if and
only if
\begin{equation}
|c_1|\leq \sum_{i=2}^r |c_i|.
\end{equation}
\end{theorem}

The spectral decomposition of $\rho^{(N)}_2$ is given by
\begin{eqnarray*}
\rho^{(N)}_2 &=& |\Psi_+\bracket \Psi_+| + |\Psi_-\bracket \Psi_-|
+\nonumber\\
&&|\Psi_{kl}\bracket \Psi_{kl}| + |\Psi_{k-l}\bracket
\Psi_{k-l}|+\nonumber\\
&&|\Psi_{-kl}\bracket \Psi_{-kl}| + |\Psi_{-k-l}\bracket
\Psi_{-k-l}|
\end{eqnarray*}
where $|\Psi_+\rangle =\sqrt{\frac{a_+}{5+a_+}}|\psi_+\rangle$,
$|\Psi_-\rangle =\sqrt{\frac{1}{5+a_+}}|\psi_-\rangle$ and
$|\Psi_{\pm k,\pm l}\rangle = \sqrt{\frac{1}{5+a_+}}|\pm k, \pm
l\rangle$. Here $|\psi_{\pm}\rangle=\frac{1}{\sqrt{2}}(|k,-k\rangle
\pm |l,-l\rangle)$ and $a_+ = (2M-N-1)/(N-1)$. Defining $\gamma^2
=1/(5+a_+),$ one obtains

\begin{equation}
C = \gamma^2\left(
              \begin{array}{cccccc}
                a_+ & 0 & 0 & 0 & 0 & 0 \\
                0 & -1 & 0 & 0 & 0 & 0 \\
                0 & 0 & 0 & -1 & 0 & 0 \\
                0 & 0 & -1 & 0 & 0 & 0 \\
                0 & 0 & 0 & 0 & 0 & 1 \\
                0 & 0 & 0 & 0 & 1 & 0 \\
              \end{array}
            \right),
\end{equation}
with spectrum $\spec (C) = \gamma^2 \{a_+,1, 1, -1, -1, -1\}$. For
$M \leq N$ the state $|\Psi_{BCS}^{(N)}\rangle$ is separable, so we
can take $M >N$. Hence,  $a_+\gamma^2$ is the eigenvalue with
biggest absolute value. According to Thm. \ref{ent-of-particles-4},
the reduced state in the subspace of the four modes is entangled iff
$|c_1|\leq \sum_{i=2}^r |c_i|$. For our example, this holds iff
$M>3N-2$.
\end{proof}
